\RequirePackage{fix-cm}
\documentclass[preprint,12pt,a4paper]{elsarticle}

\usepackage[T1]{fontenc}
\usepackage[utf8]{inputenc}
\usepackage[english]{babel}
\usepackage{lmodern}
\usepackage{microtype}
\usepackage{amsmath,amssymb,amsfonts,amsthm,mathtools}
\usepackage{booktabs,array,multirow,longtable}
\usepackage{graphicx,rotating,caption}
\usepackage{enumitem}
\usepackage{hyperref}
\usepackage[nameinlink,noabbrev]{cleveref}
\usepackage{url}
\usepackage[pass]{geometry}
\allowdisplaybreaks[2]
\journal{Finite Fields and Their Applications}
\biboptions{sort&compress}
\abstracttitle{Abstract}
\keywordtitle{Keywords}

\newtheorem{theorem}{Theorem}[section]
\newtheorem{lemma}[theorem]{Lemma}
\newtheorem{corollary}[theorem]{Corollary}
\newtheorem{proposition}[theorem]{Proposition}

\newcommand{\Tr}{\operatorname{Tr}}
\newcommand{\N}{\operatorname{N}}
\newcommand{\GRS}{\operatorname{GRS}}

\newcommand{\qcode}[3]{\left[\!\left[#1,#2,#3\right]\!\right]_q}

\hypersetup{
  hidelinks,
  pdftitle={Quantum MDS codes from complements of unions of finite-field subsets},
  pdfauthor={Naihong Hu and Hong Ji}
}

\newlength{\knownwidth}
\newlength{\knownleft}
\newlength{\knownright}

\begin{document}
\begin{frontmatter}

\title{Quantum MDS codes from complements\\of unions of finite-field subsets}

\author[aff1]{Naihong Hu}
\author[aff1]{Hong Ji\corref{cor1}}
\ead{52280155008@stu.ecnu.edu.cn}
\cortext[cor1]{Corresponding author.}
\address[aff1]{School of Mathematical Sciences, MOE Key Laboratory of Mathematics and Engineering Applications $\&$ Shanghai Key Laboratory of PMMP, East China Normal University, Shanghai 200241, China}

\begin{abstract}
Let $q$ be an odd prime power. We use complements of unions of subsets of $\mathbb F_{q^2}$ as locator sets and establish a sufficient condition under which a generalized Reed--Solomon (GRS) code is Hermitian self-orthogonal. Using cosets of multiplicative subgroups and sets with prescribed trace or norm values, we construct five families of Hermitian self-orthogonal GRS codes over $\mathbb F_{q^2}$. The Hermitian construction then yields five corresponding families of $q$-ary quantum maximum-distance-separable (MDS) codes. Under suitable parameter conditions, these quantum codes have minimum distances greater than $q/2+1$. By comparing codes of the same length, we give conditions under which our codes have strictly larger minimum distances than those obtainable from several previously known constructions based on trace maps, linear transformations, and cosets of multiplicative subgroups, either directly or via the propagation rule. We further show that such improvements occur for infinitely many values of $q$.
\end{abstract}

\begin{keyword}
Hermitian self-orthogonal code \sep Generalized Reed--Solomon code \sep Finite field \sep Trace map \sep Norm map \sep Quantum MDS code
\MSC[2020] 11T71 \sep 94B05 \sep 81P70
\end{keyword}

\end{frontmatter}

\section{Introduction}\label{sec:introduction}

Quantum error-correcting codes protect quantum information against errors arising during storage, transmission, and computation. Shor's nine-qubit code \cite{Shor1995} demonstrated that arbitrary errors acting on a single qubit can be corrected. Subsequently, the constructions of Calderbank and Shor \cite{CalderbankShor1996} and Steane \cite{Steane1996}, together with the stabilizer construction based on additive codes over $\mathbb F_4$ \cite{CalderbankRainsShorSloane1998}, established a close connection between classical and quantum coding theory. This connection was extended to general finite fields in \cite{AshikhminKnill2001,KetkarEtAl2006}, providing an important approach to constructing quantum codes with good parameters.

We use $\qcode{n}{k}{d}$ to denote a $q$-ary quantum code of length $n$, size $q^k$, and minimum distance $d$. Its parameters satisfy the quantum Singleton bound \cite{KetkarEtAl2006}
\begin{equation}\label{eq:intro-singleton}
 2d\le n-k+2.
\end{equation}
A quantum code attaining equality is called a quantum maximum-distance-separable (MDS) code. Such a code has the largest possible minimum distance for its length and dimension. Quantum MDS codes of length at most $q+1$ have been constructed for a broad range of parameters \cite{GrasslBethRoetteler2004,JinXing2014}. For lengths greater than $q+1$, the attainable parameters depend more strongly on the particular construction. In particular, constructing quantum MDS codes with new lengths and minimum distances greater than $q/2+1$ remains an important problem \cite{FangFu2018,GuoLiLiu2021,WanZhengZhu2025}.

The Hermitian construction reduces this problem to that of constructing classical self-orthogonal MDS codes. If an $[n,k,n-k+1]_{q^2}$ MDS code $C$ satisfies $C\subseteq C^{\perp_H}$, then there exists a $q$-ary quantum MDS code with parameters $\qcode{n}{n-2k}{k+1}$ \cite{AshikhminKnill2001,KetkarEtAl2006}. Thus, for a fixed length, increasing the dimension of the classical self-orthogonal code increases the minimum distance of the resulting quantum code. Generalized Reed--Solomon (GRS) codes are particularly well suited to this construction: their MDS property follows from polynomial evaluation, so the main task is to choose code locators and column multipliers that ensure Hermitian self-orthogonality \cite{MacWilliamsSloane1977}. The algebraic structure of the locator set is therefore central to the construction.

Li, Xing, and Wang \cite{LiXingWang2008} established a unified framework for constructing quantum MDS codes from GRS codes. Jin et al.\ \cite{JinLingLuoXing2010,JinXing2014} further developed this approach. A key step is to reduce the choice of column multipliers to finding a solution in $(\mathbb F_q^{*})^n$ of a homogeneous system over $\mathbb F_{q^2}$. Sufficient conditions ensuring the existence of such solutions were subsequently obtained in \cite{WanZhengZhu2025}. By combining these methods with algebraically structured locator sets, researchers have constructed $q$-ary quantum MDS codes whose minimum distances can exceed $q/2+1$ \cite{FangFu2018,FangFu2019}. Restrictions on the attainable lengths have also been relaxed. For example, Guo, Li, and Liu \cite{GuoLiLiu2021} constructed quantum MDS codes of lengths $m(q-1)$ and $s(q+1)$, thereby relaxing divisibility conditions on $m$ and $s$ imposed in earlier constructions. Subsequent work produced quantum MDS codes whose lengths satisfy particular congruence conditions \cite{WanLianZhu2025} or have product forms \cite{CampionHernandoMcGuire2025}, further enlarging the range of attainable lengths.

Constructions based on unions of cosets of multiplicative subgroups have yielded quantum MDS codes with flexible lengths \cite{FangLuo2020,JinLuoFangQu2022}. Related work has also produced both quantum MDS codes and entanglement-assisted quantum MDS codes \cite{TianLiWuChen2024}. Fang, Wen, and Fu \cite{FangWenFu2024} gave sufficient conditions for constructing Hermitian self-orthogonal GRS codes on the union of two locator sets. By combining trace maps with finite-field subspaces in one construction and with norm maps in the other, they obtained $q$-ary quantum MDS codes of lengths $(s+t)q-st$ and $1+sq+t(q+1)-N_{s,t}$, respectively. The quantities $st$ and $N_{s,t}$ record the sizes of the corresponding intersections, illustrating how intersections affect the code length. Li, Liu, and Jiang \cite{LiLiuJiang2026} used affine and projective linear transformations to construct quantum MDS codes whose minimum distances can exceed $q/2+1$.

Several general criteria for selecting code locators are also known. Ball and Vilar \cite{BallVilar2022} characterized the Hermitian self-orthogonality of truncated GRS codes by a polynomial condition. In the Euclidean self-dual setting, Meng, Fang, Fu, Zhou, and Gu \cite{MengFangFuZhouGu2026} constructed MDS codes from multiple intersecting subsets of finite fields. These results address different orthogonality conditions or different types of locator sets. Here we ask how $\mathbb F_q^{*}$-cosets and sets with prescribed trace or norm values can be used to specify the code locators and the equations satisfied by the column multipliers.

In this paper, we take $A=X\setminus B$ as the locator set, where $X$ is either $\mathbb F_{q^2}$ or $\mathbb F_{q^2}^{*}$ and $B$ is a union of selected subsets. This complement structure allows the Lagrange coefficients on $A$ to be expressed in terms of the vanishing polynomial of $B$. By studying the product of the vanishing polynomials of the individual subsets and its behavior under the Frobenius automorphism $x\mapsto x^q$, we specify the norms of the column multipliers and derive the sufficient condition for Hermitian self-orthogonality in Theorem~\ref{thm:general-construction}. The condition applies to any finite collection of subsets satisfying the stated hypotheses and permits nonempty intersections among them. The resulting dimension bound does not explicitly depend on the intersection sizes, whereas the code length depends on them through the cardinality of the union $B$.

Applying this condition, we obtain five families of Hermitian self-orthogonal GRS codes. The first combines sets with prescribed norm values and $\mathbb F_q^{*}$-cosets. The second combines sets with prescribed trace values and $\mathbb F_q^{*}$-cosets, whereas the third combines sets with prescribed trace values and sets with prescribed norm values. The fourth and fifth combine two families of sets with prescribed trace values with $\mathbb F_q^{*}$-cosets and sets with prescribed norm values, respectively.

The Hermitian construction then gives five families of quantum MDS codes. Their lengths have the common form $n=q^2-Sq+\theta$, and codes exist for every minimum distance in the range $2\le d\le d_*=q-S+1$. Here $S$ is determined by the cardinalities of the selected parameter sets, whereas $\theta$ depends on these cardinalities and the corresponding intersection sizes. When the numbers of selected trace values, norm values, and cosets are fixed, the constructions give codes of length $q^2-O(q)$ whose minimum distances can be as large as $q-O(1)$; in particular, $d_*>q/2+1$ for all sufficiently large $q$.

Using this common parameterization, we compare our codes at the same length with those obtainable from previously known constructions, either directly or via the propagation rule. For the trace--subspace and trace--norm constructions in \cite{FangWenFu2024}, we derive upper bounds on the attainable minimum distances and sufficient conditions under which each of our five families has a strictly larger minimum distance. For the constructions based on linear transformations in \cite{LiLiuJiang2026}, we determine the largest minimum distance attainable in the range $(q^2-1)/2<n<q^2$ from the first two constructions, including all codes derived from them via the propagation rule. Applying these comparison formulas to our second family yields a two-parameter family of length $(q-R)(q-b)$ together with conditions for a strict improvement in minimum distance; see Corollary~\ref{cor:second-comparison}. In addition, Corollary~\ref{cor:infinite-comparison} exhibits one infinite subfamily from each of our five constructions and gives its gain in minimum distance over the linear-transformation constructions.

For previously known constructions based on cosets of multiplicative subgroups, we combine congruence conditions on the original lengths with the propagation rule to derive fixed-length upper bounds on the minimum distances of the codes listed as R5--R16, including all codes derived from them via the propagation rule. We also give two sufficient conditions for strict improvement. Corollary~\ref{cor:second-coset-family} explicitly selects the sets required for our second construction and proves that, under the stated conditions, the resulting two-parameter family has minimum distance strictly greater than these bounds. The original lengths in R17--R20 are at most $(q^2-1)/2$, and the propagation rule does not increase the length; hence neither these families nor codes derived from them via the propagation rule can have the lengths covered by that corollary. The parameters of R3--R20 are summarized in \ref{app:known}.

The remainder of the paper is organized as follows. Section~\ref{sec:preliminaries} recalls GRS codes and a criterion for Hermitian self-orthogonality. Section~\ref{sec:classical} presents the general construction, the required intersection formulas and vanishing polynomials, and five families of Hermitian self-orthogonal GRS codes. Section~\ref{sec:quantum} gives the parameter ranges of the corresponding quantum MDS codes and compares them at the same length with codes obtainable from previously known constructions based on trace maps, linear transformations, and cosets of multiplicative subgroups. Section~\ref{sec:conclusion} concludes the paper. \ref{app:known} summarizes the previously known families used in the comparisons and several other related results.

\section{Preliminaries}\label{sec:preliminaries}

Throughout this paper, $q$ is an odd prime power. A linear $[n,k,d]_{q^2}$ code is a $k$-dimensional subspace of $\mathbb F_{q^2}^{n}$ with minimum Hamming distance $d$. Write $\Tr(x)=x+x^q$ and $\N(x)=x^{q+1}$ for the trace and norm maps of the extension $\mathbb F_{q^2}/\mathbb F_q$, respectively. The multiplicative group $\mathbb F_{q^2}^{*}$ is cyclic, and the norm map is a surjection onto $\mathbb F_q^{*}$ with kernel
\[
 U=\{z\in\mathbb F_{q^2}^{*}:z^{q+1}=1\},\qquad |U|=q+1.
\]
The map $x\mapsto x^{q-1}$ is a surjection from $\mathbb F_{q^2}^{*}$ onto $U$ with kernel $\mathbb F_q^{*}$. These standard facts about finite fields may be found in \cite[Chapter~2]{LidlNiederreiter1997}. We adopt the conventions that an empty product is $1$ and $\deg 0=-\infty$.

For $x,y\in\mathbb F_{q^2}^{n}$, the Euclidean and Hermitian inner products are defined, respectively, by
\[
 \langle x,y\rangle_E=\sum_{i=1}^{n}x_i y_i,
 \qquad
 \langle x,y\rangle_H=\sum_{i=1}^{n}x_i y_i^q.
\]
The corresponding duals of a linear code $C\subseteq\mathbb F_{q^2}^{n}$ are
\[
\begin{aligned}
 C^{\perp_E}&=\{x\in\mathbb F_{q^2}^{n}:\langle x,c\rangle_E=0,
                    \ \text{for all }c\in C\},\\
 C^{\perp_H}&=\{x\in\mathbb F_{q^2}^{n}:\langle x,c\rangle_H=0,
                    \ \text{for all }c\in C\}.
\end{aligned}
\]
The code $C$ is called Euclidean self-orthogonal if $C\subseteq C^{\perp_E}$ and Hermitian self-orthogonal if $C\subseteq C^{\perp_H}$. Write $x^q=(x_1^q,\ldots,x_n^q)$. Since $\langle x,c\rangle_H^q=\langle x^q,c\rangle_E$, we have $x\in C^{\perp_H}$ if and only if $x^q\in C^{\perp_E}$.

Let $A=\{\alpha_1,\ldots,\alpha_n\}\subseteq\mathbb F_{q^2}$ be a set of $n$ distinct elements. For $\boldsymbol v=(v_1,\ldots,v_n)\in(\mathbb F_{q^2}^{*})^n$ and $1\le k\le n$, define
\[
 \GRS_k(A,\boldsymbol v)=\left\{(v_1f(\alpha_1),\ldots,v_nf(\alpha_n)):
 f\in\mathbb F_{q^2}[x],\ \deg f\le k-1\right\}.
\]
This is an \mbox{$[n,k,n-k+1]_{q^2}$} MDS code \cite{MacWilliamsSloane1977}. The elements of $A$ are called code locators (or evaluation points), and the components of $\boldsymbol v$ are called column multipliers. We also index the multipliers by the code locators, writing $v_\alpha$ for the multiplier corresponding to $\alpha\in A$.

For a finite set $S\subseteq\mathbb F_{q^2}$, let $F_S(x)=\prod_{\beta\in S}(x-\beta)$ be its monic vanishing polynomial. For each $\alpha\in A$, define the corresponding Lagrange coefficient by
\[
 w_\alpha=\frac{1}{F_A'(\alpha)}
 =\prod_{\substack{\beta\in A\\\beta\ne\alpha}}(\alpha-\beta)^{-1}.
\]
The Lagrange interpolation formula gives $\sum_{\alpha\in A}w_\alpha H(\alpha)=0$ whenever $\deg H\le n-2$. Consequently, the Euclidean dual of a GRS code is given by \cite{MacWilliamsSloane1977,FangFu2018}
\[
 \GRS_k(A,\boldsymbol v)^{\perp_E}
 =\left\{\bigl(w_\alpha v_\alpha^{-1}g(\alpha)\bigr)_{\alpha\in A}:
 g\in\mathbb F_{q^2}[x],\ \deg g\le n-k-1\right\}.
\]
Combining this formula with the relation between the Euclidean and Hermitian duals gives the following criterion; see \cite[Lemma~6]{FangFu2018} and \cite{WanZhengZhu2025}.

\begin{lemma}\label{lem:hermitian-dual}
Let $c_f=(v_\alpha f(\alpha))_{\alpha\in A}\in\GRS_k(A,\boldsymbol v)$. Then $c_f\in\GRS_k(A,\boldsymbol v)^{\perp_H}$ if and only if there exists $g\in\mathbb F_{q^2}[x]$ with $\deg g\le n-k-1$ such that, for every $\alpha\in A$,
\[
 v_\alpha^{q+1}f(\alpha)^q=w_\alpha g(\alpha).
\]
\end{lemma}

For $f(x)=\sum_j f_jx^j$, define $f^{[q]}(x)=\sum_j f_j^q x^{qj}$. Then $f^{[q]}(\alpha)=f(\alpha)^q$, and $\deg f^{[q]}=q\deg f$ whenever $f\ne0$. Taking $g=hf^{[q]}$ in Lemma~\ref{lem:hermitian-dual} yields the following criterion.

\begin{proposition}\label{prop:polynomial-criterion}
Let $k$ be a positive integer. Suppose that there exists $h\in\mathbb F_{q^2}[x]$ such that $\deg h+q(k-1)\le n-k-1$ and
\[
 w_\alpha h(\alpha)\in\mathbb F_q^{*}
\]
for every $\alpha\in A$. Then there exists $\boldsymbol v=(v_\alpha)_{\alpha\in A}\in (\mathbb F_{q^2}^{*})^n$
such that $\GRS_k(A,\boldsymbol v)$ is a Hermitian self-orthogonal \mbox{$[n,k,n-k+1]_{q^2}$} MDS code.
\end{proposition}

\begin{proof}
The condition $w_\alpha h(\alpha)\in\mathbb F_q^{*}$ implies $h\ne0$, so $\deg h\ge0$. Together with $k\ge1$ and the degree condition, this gives $1\le k\le n-1$.

For each $\alpha \in A$, the surjectivity of the norm map onto $\mathbb F_q^{*}$ allows us to choose $v_\alpha \in \mathbb F_{q^2}^{*}$ such that $v_\alpha^{q+1}=w_\alpha h(\alpha)$. Take any $f\in\mathbb F_{q^2}[x]$ with $\deg f\le k-1$ and set $g=hf^{[q]}$. Then $\deg g\le n-k-1$, and for every $\alpha\in A$,
\[
 v_\alpha^{q+1}f(\alpha)^q
 =w_\alpha h(\alpha)f(\alpha)^q=w_\alpha g(\alpha).
\]
By Lemma~\ref{lem:hermitian-dual}, every codeword of $\GRS_k(A,\boldsymbol v)$ lies in its Hermitian dual. Hence $\GRS_k(A,\boldsymbol v)\subseteq\GRS_k(A,\boldsymbol v)^{\perp_H}$.
\end{proof}

\section{Hermitian self-orthogonal GRS codes}\label{sec:classical}\label{sec:general}\label{sec:families}

In this section, we construct Hermitian self-orthogonal GRS codes whose locator sets are complements of unions of subsets of $\mathbb F_{q^2}$ or $\mathbb F_{q^2}^{*}$. We first establish a sufficient condition for Hermitian self-orthogonality and then apply it to $\mathbb F_q^{*}$-cosets and sets with prescribed trace or norm values.

\subsection{A general construction}

When a locator set is the complement of a subset $B$ in an ambient finite set $X$, its Lagrange coefficients can be expressed in terms of the vanishing polynomial of $B$.

\begin{lemma}\label{lem:complement-coefficients}
Let $X\in\{\mathbb F_{q^2},\mathbb F_{q^2}^{*}\}$ and $B\subsetneq X$, and put $A=X\setminus B$. Then, for every $\alpha\in A$,
\[
 w_\alpha=
 \begin{cases}
 -F_B(\alpha),&X=\mathbb F_{q^2},\\
 -\alpha F_B(\alpha),&X=\mathbb F_{q^2}^{*}
 \end{cases}.
\]
\end{lemma}

\begin{proof}
Fix $\alpha\in A$. If $X=\mathbb F_{q^2}$, then $x^{q^2}-x=F_A(x)F_B(x)$. Differentiating and evaluating at $\alpha$ gives $-1=F_A'(\alpha)F_B(\alpha)$. If $X=\mathbb F_{q^2}^{*}$, differentiating $x^{q^2-1}-1=F_A(x)F_B(x)$ gives $-\alpha^{-1}=F_A'(\alpha)F_B(\alpha)$. In both cases, the result follows from $w_\alpha=1/F_A'(\alpha)$.
\end{proof}

When $B$ is a union of subsets, the next lemma expresses $F_B$ in terms of the product of their vanishing polynomials.

\begin{lemma}\label{lem:union-vanishing}
Let $B_1,\ldots,B_s\subseteq\mathbb F_{q^2}$, let $B=\bigcup_iB_i$, and put $P_i=c_iF_{B_i}$, where $c_i\in\mathbb F_{q^2}^{*}$. Write $P=\prod_iP_i$ and $c_0=\prod_ic_i$. Then there exists a monic polynomial $J\in\mathbb F_{q^2}[x]$ such that
\[
 P=c_0JF_B,
 \qquad \deg J=\deg P-|B|.
\]
\end{lemma}

\begin{proof}
Since $P$ vanishes at every element of $B$, we have $F_B\mid P$. Set $J=P/(c_0F_B)$. As the leading coefficient of $P$ is $c_0$, the polynomial $J$ is monic. Comparing degrees completes the proof.
\end{proof}

For use in the following theorem, let $B_1,\ldots,B_s\subseteq X$, where $X\in\{\mathbb F_{q^2},\mathbb F_{q^2}^{*}\}$. Retain the notation $B,P,c_0,J$ from Lemma~\ref{lem:union-vanishing}, assume $B\ne X$, and set $A=X\setminus B$ and $n=|A|$. Define $\varepsilon_X=0$ when $X=\mathbb F_{q^2}$ and $\varepsilon_X=1$ when $X=\mathbb F_{q^2}^{*}$. Proposition~\ref{prop:polynomial-criterion} can then be applied to the complement $A$.

\begin{theorem}\label{thm:general-construction}
With the notation above, let $m\ge\varepsilon_X$ be an integer and put $G(x)=x^mP(x)$. Suppose that there exists $\gamma\in U$ such that, for every $\alpha\in A$,
\begin{equation}\label{eq:general-frobenius}
 G(\alpha)\ne0,\qquad G(\alpha)^q=\gamma G(\alpha).
\end{equation}
Then, for every positive integer $k$ satisfying
\begin{equation}\label{eq:general-degree}
 (q+1)k\le q^2+q-1-m-\deg P
\end{equation}
there exists $\boldsymbol v=(v_\alpha)_{\alpha\in A}\in(\mathbb F_{q^2}^{*})^n$ such that $\GRS_k(A,\boldsymbol v)$ is a Hermitian self-orthogonal \mbox{$[n,k,n-k+1]_{q^2}$} MDS code.
\end{theorem}

\begin{proof}
Since $x\mapsto x^{q-1}$ maps $\mathbb F_{q^2}^{*}$ onto $U$, we may choose $\eta\in\mathbb F_{q^2}^{*}$ with $\eta^{q-1}=\gamma^{-1}$. Define $h(x)=c_0\eta x^{m-\varepsilon_X}J(x)$. The exponent is nonnegative, so $h\in\mathbb F_{q^2}[x]$. When $m=0$, the factor $x^0$ is the constant polynomial $1$. When $m>0$, the nonvanishing of $G$ on $A$ implies $0\notin A$.

Fix $\alpha\in A$. By \eqref{eq:general-frobenius},
\[
 (\eta G(\alpha))^q=\eta^q\gamma G(\alpha)=\eta G(\alpha).
\]
Together with $G(\alpha)\ne0$, Lemmas~\ref{lem:complement-coefficients} and~\ref{lem:union-vanishing} give
\[
 w_\alpha h(\alpha)=-\eta G(\alpha)\in\mathbb F_q^{*}.
\]

It remains to verify the degree condition. Using $n=|X|-|B|$, $\deg J=\deg P-|B|$, and $|X|+\varepsilon_X=q^2$, we obtain
\[
\begin{aligned}
 \deg h&=m-\varepsilon_X+\deg J,\\
 n+q-1-\deg h&=q^2+q-1-m-\deg P.
\end{aligned}
\]
Thus \eqref{eq:general-degree} is equivalent to $\deg h+q(k-1)\le n-k-1$. The result now follows from Proposition~\ref{prop:polynomial-criterion}.
\end{proof}

We next specialize Theorem~\ref{thm:general-construction} to the case $P(x)=H(x)Q(x)$, where $H(\alpha)\in\mathbb F_q^{*}$ for every $\alpha\in A$ and $Q$ is the monic vanishing polynomial of a union of $b$ distinct $\mathbb F_q^{*}$-cosets. Taking $m=b$ gives the following result.

\begin{lemma}\label{lem:q-power-calculation}
Retain the notation and assumptions of Theorem~\ref{thm:general-construction}, and suppose that $0\notin A$. Let $D\subseteq U$, set $b=|D|\ge1$, and define
\[
 Q(x)=\prod_{\delta\in D}(x^{q-1}-\delta),
 \qquad C_D=(-1)^b\left(\prod_{\delta\in D}\delta\right)^{-1}.
\]
Suppose that $P(x)=H(x)Q(x)$ for some $H\in\mathbb F_{q^2}[x]$ satisfying $H(\alpha)\in\mathbb F_q^{*}$ for every $\alpha\in A$.
Then $C_D\in U$, and for every positive integer $k$ satisfying
\begin{equation}\label{eq:specialized-degree}
 (q+1)k\le q^2+q-1-bq-\deg H
\end{equation}
there exists $\boldsymbol v=(v_\alpha)_{\alpha\in A}\in(\mathbb F_{q^2}^{*})^n$ such that $\GRS_k(A,\boldsymbol v)$ is a Hermitian self-orthogonal \mbox{$[n,k,n-k+1]_{q^2}$} MDS code.
\end{lemma}

\begin{proof}
Take $m=b$ in Theorem~\ref{thm:general-construction} and put
\[
 G(x)=x^bP(x)=x^bH(x)Q(x).
\]
Then $m=b\ge1\ge\varepsilon_X$. Since $P=\prod_i c_iF_{B_i}$ and $A\cap B_i=\varnothing$, we have $P(\alpha)\ne0$ for every $\alpha\in A$. As $0\notin A$, it follows that $G(\alpha)\ne0$.

The inclusions $-1\in U$ and $D\subseteq U$ imply $C_D\in U$. Fix $\alpha\in A$. Since $\alpha\ne0$, we have $\alpha^{q(q-1)}=\alpha^{1-q}$. Moreover, $\delta^q=\delta^{-1}$ for every $\delta\in D$. Hence
\[
\begin{aligned}
 Q(\alpha)^q
 &=\prod_{\delta\in D}(\alpha^{1-q}-\delta^{-1})\\
 &=\prod_{\delta\in D}\left(-\frac{\alpha^{q-1}-\delta}{\alpha^{q-1}\delta}\right)\\
 &=C_D\alpha^{-b(q-1)}Q(\alpha).
\end{aligned}
\]
Using $H(\alpha)^q=H(\alpha)$, we obtain
\[
\begin{aligned}
 G(\alpha)^q
 &=\alpha^{bq}H(\alpha)^qQ(\alpha)^q\\
 &=C_D\alpha^bH(\alpha)Q(\alpha)
 =C_DG(\alpha).
\end{aligned}
\]
Thus the Frobenius condition~\eqref{eq:general-frobenius} in Theorem~\ref{thm:general-construction} holds with $\gamma=C_D$. Since $\deg Q=b(q-1)$, we also have
\[
 \deg P=\deg H+b(q-1).
\]
Consequently, $m+\deg P=bq+\deg H$, so the degree condition~\eqref{eq:general-degree} in that theorem is exactly \eqref{eq:specialized-degree}. The result follows from Theorem~\ref{thm:general-construction}.
\end{proof}

\subsection{Intersection formulas and vanishing polynomials}

We now give intersection formulas and monic vanishing polynomials for the subsets used in our constructions. These formulas will be used to compute the code lengths and to choose the polynomial factors in Theorem~\ref{thm:general-construction}.

Fix a primitive element $g$ of $\mathbb F_{q^2}^{*}$ and write $\rho=g^{q-1}$ and $\xi=g^{q+1}$. Then $U=\langle\rho\rangle$ and $\mathbb F_q^{*}=\langle\xi\rangle$. For $\delta\in U$ and $\nu\in\mathbb F_q^{*}$, define
\[
 M_\delta=\{x\in\mathbb F_{q^2}^{*}:x^{q-1}=\delta\},
 \qquad
 \mathcal U_\nu=\{x\in\mathbb F_{q^2}^{*}:\N(x)=\nu\}.
\]
The following lemma describes these sets as cosets and gives their intersection numbers.

\begin{lemma}\label{lem:coset-intersections}
With the notation above, the following statements hold.
\begin{enumerate}[label=\textup{(\roman*)}]
\item If $\delta=\rho^j$, then $M_\delta=g^j\mathbb F_q^{*}$. As $\delta$ ranges over $U$, the sets $M_\delta$ form a partition of $\mathbb F_{q^2}^{*}$, and each has cardinality $q-1$.
\item If $\nu=\xi^i$, then $\mathcal U_\nu=g^iU$. As $\nu$ ranges over $\mathbb F_q^{*}$, the sets $\mathcal U_\nu$ form a partition of $\mathbb F_{q^2}^{*}$, and each has cardinality $q+1$.
\item For any integers $i,j$,
\[
 |\mathcal U_{\xi^i}\cap M_{\rho^j}|=
 \begin{cases}
 2,&i\equiv j\pmod2,\\
 0,&i\not\equiv j\pmod2.
 \end{cases}
\]
\item If $\nu$ is a square in $\mathbb F_q^{*}$, then every element of $\mathcal U_\nu$ is a square in $\mathbb F_{q^2}^{*}$. If $\nu$ is a nonsquare in $\mathbb F_q^{*}$, then every element of $\mathcal U_\nu$ is a nonsquare in $\mathbb F_{q^2}^{*}$. There are $(q-1)/2$ sets of each type.
\end{enumerate}
\end{lemma}

\begin{proof}
\begin{enumerate}[label=\textup{(\roman*)}]
\item The equation $x^{q-1}=\rho^j$ is equivalent to $x/g^j\in\mathbb F_q^{*}$, so $M_\delta=g^j\mathbb F_q^{*}$. The preimages of distinct $\delta$ are disjoint, giving the stated partition into sets of size $q-1$.

\item The equation $x^{q+1}=\xi^i$ is equivalent to $x/g^i\in U$, so $\mathcal U_\nu=g^iU$. The preimages of distinct $\nu$ are disjoint, giving the stated partition into sets of size $q+1$.

\item The subgroups $U=\langle g^{q-1}\rangle$ and $\mathbb F_q^{*}=\langle g^{q+1}\rangle$ are both contained in the subgroup of squares $S=(\mathbb F_{q^2}^{*})^2$. Moreover, $U\cap\mathbb F_q^{*}=\{1,-1\}$, since every element of the intersection has square equal to $1$. Thus
\[
 |U\mathbb F_q^{*}|=\frac{(q+1)(q-1)}{2}=|S|,
 \qquad U\mathbb F_q^{*}=S.
\]
By (i) and (ii), $\mathcal U_{\xi^i}=g^iU$ and $M_{\rho^j}=g^j\mathbb F_q^{*}$. The two cosets intersect if and only if $g^{j-i}\in U\mathbb F_q^{*}=S$, equivalently, $i\equiv j\pmod2$. If $x$ lies in the intersection, then any other element $y$ of the intersection satisfies $x^{-1}y\in U\cap\mathbb F_q^{*}=\{1,-1\}$, so $y=x$ or $y=-x$. Conversely, since $-1\in U\cap\mathbb F_q^{*}$, both $x$ and $-x$ lie in the intersection. These points are distinct because $q$ is odd and $x\ne0$. Hence the intersection is $\{x,-x\}$.

\item Write $\nu=\xi^i$. Since $\mathcal U_\nu=g^iU$ and $U\subseteq S$, all elements of $\mathcal U_\nu$ are squares in $\mathbb F_{q^2}^{*}$ if and only if $i$ is even. As $\mathbb F_q^{*}=\langle\xi\rangle$, this is equivalent to $\nu$ being a square in $\mathbb F_q^{*}$. There are $(q-1)/2$ squares and $(q-1)/2$ nonsquares in $\mathbb F_q^{*}$, which proves the assertion.
\end{enumerate}
\end{proof}

The trace map is a surjective $\mathbb F_q$-linear map from the two-dimensional vector space $\mathbb F_{q^2}$ onto $\mathbb F_q$, with kernel of cardinality $q$. We may therefore choose a nonzero $\omega$ such that $\omega^q=-\omega$ \cite[Chapter~2]{LidlNiederreiter1997}. Since $q$ is odd, $\omega\notin\mathbb F_q$, and hence $\{1,\omega\}$ is a basis. Put $\Delta=\omega^2$. Then $\Delta\in\mathbb F_q^{*}$ and $\Delta^{(q-1)/2}=\omega^{q-1}=-1$, so $\Delta$ is a nonsquare. For $x=s+t\omega$, where $s,t\in\mathbb F_q$, the identity $x^q=s-t\omega$ gives
\[
 \Tr(x)=2s,\qquad \Tr(-\omega x)=-2\Delta t,
 \qquad \N(x)=s^2-\Delta t^2.
\]

Fix this choice of $\omega$ for the remainder of the paper. For $c,e\in\mathbb F_q$, put
\[
\begin{aligned}
 L_c&=\{x\in\mathbb F_{q^2}:\Tr(x)=c\},\\
 K_e&=\{x\in\mathbb F_{q^2}:\Tr(-\omega x)=e\}.
\end{aligned}
\]
Here $\Tr(-\omega x)=\omega(x^q-x)$. Within each family, distinct parameter values give pairwise disjoint sets.

\begin{proposition}\label{prop:intersection-formulas}
For all $c,e\in\mathbb F_q$, the sets $L_c$ and $K_e$ each have cardinality $q$ and satisfy
\[
 L_c=\frac c2+\omega\mathbb F_q,\qquad
 K_e=-\frac{e}{2\omega}+\mathbb F_q.
\]
Their intersection consists of the single point
\[
 x_{c,e}=\frac12\left(c-\frac e\omega\right),
 \qquad \N(x_{c,e})=\frac{c^2-e^2/\Delta}{4}.
\]
For every $\delta\in U\setminus\{-1\}$, we have $L_0\cap M_\delta=\varnothing$, and for every $c\in\mathbb F_q^{*}$,
\[
 L_c\cap M_\delta=\left\{\frac{c}{\delta+1}\right\}.
\]
For every $\delta\in U\setminus\{1\}$, we have $K_0\cap M_\delta=\varnothing$, and for every $e\in\mathbb F_q^{*}$,
\[
 K_e\cap M_\delta=\left\{\frac{e}{\omega(\delta-1)}\right\}.
\]
\end{proposition}

\begin{proof}
In the coordinate representation $x=s+t\omega$, the condition $\Tr(x)=c$ fixes $s=c/2$ while $t$ varies over $\mathbb F_q$, and the condition $\Tr(-\omega x)=e$ fixes $t=-e/(2\Delta)$ while $s$ varies over $\mathbb F_q$. This gives the stated descriptions and cardinalities of $L_c$ and $K_e$. Fixing both coordinates gives the unique intersection point $x_{c,e}$ and its norm.

For $x\in M_\delta$, the relation $x^q=\delta x$ yields
\[
 \Tr(x)=(\delta+1)x,\qquad
 \Tr(-\omega x)=\omega(\delta-1)x.
\]
If $c\ne0$ and $\delta\ne-1$, the only possible intersection point is $x=c/(\delta+1)$. This point indeed belongs to $M_\delta$, since
\[
 x^q=\frac{c}{\delta^{-1}+1}=\frac{c\delta}{1+\delta}=\delta x.
\]
Similarly, if $e\ne0$ and $\delta\ne1$, the only possible intersection point is $x=e/[\omega(\delta-1)]$, and
\[
 x^q=\frac{e}{(-\omega)(\delta^{-1}-1)}=\delta x.
\]
The displayed points also satisfy their respective trace conditions. If the prescribed trace value is zero, the nonzero coefficients above force $x=0$. Since $0\notin M_\delta$, the corresponding intersections are empty.
\end{proof}

The first trace--norm intersection formula below agrees with \cite[Lemma~8]{FangWenFu2024}. The same coordinate representation also gives the intersection number of $K_e$ and $\mathcal U_\nu$.

\begin{lemma}\label{lem:trace-norm-intersections}
Let $\chi$ be the quadratic character of $\mathbb F_q$, extended by $\chi(0)=0$. For all $c,e\in\mathbb F_q$ and $\nu\in\mathbb F_q^{*}$,
\[
\begin{aligned}
 |L_c\cap\mathcal U_\nu|&=1-\chi(c^2-4\nu),\\
 |K_e\cap\mathcal U_\nu|&=1+\chi\left(\nu+\frac{e^2}{4\Delta}\right).
\end{aligned}
\]
\end{lemma}

\begin{proof}
Every element of $L_c$ is uniquely expressible as $c/2+t\omega$, where $t\in\mathbb F_q$, and has norm $c^2/4-\Delta t^2$. Thus the first intersection size equals the number of solutions of
\[
 t^2=\frac{c^2-4\nu}{4\Delta}
\]
in $\mathbb F_q$, namely $1+\chi((c^2-4\nu)/(4\Delta))$. Since $\chi(\Delta)=-1$, this equals $1-\chi(c^2-4\nu)$. In particular, the intersection size is 1, 2, or 0 according as $c^2-4\nu$ is zero, a nonsquare, or a nonzero square.

Similarly, every element of $K_e$ is uniquely expressible as $s-e/(2\omega)$, where $s\in\mathbb F_q$, and has norm $s^2-e^2/(4\Delta)$. The second intersection size is therefore the number of solutions of
\[
s^2=\nu+\frac{e^2}{4\Delta}\]
in $\mathbb F_q$, giving the stated formula.
\end{proof}

We next record the monic vanishing polynomials of these sets, which will provide explicit choices of the factors $P_i$ in Theorem~\ref{thm:general-construction}.

\begin{lemma}\label{lem:basic-vanishing-polynomials}
For $c,e\in\mathbb F_q$, $\nu\in\mathbb F_q^{*}$, and $\delta\in U$, the corresponding monic vanishing polynomials are
\begin{equation}\label{eq:basic-vanishing-polynomials}
\begin{aligned}
 F_{L_c}(x)&=x^q+x-c,&
 F_{\mathcal U_\nu}(x)&=x^{q+1}-\nu,\\
 F_{K_e}(x)&=x^q-x-e/\omega,&
 F_{M_\delta}(x)&=x^{q-1}-\delta.
\end{aligned}
\end{equation}
In particular, all four polynomials have only simple roots.
\end{lemma}

\begin{proof}
Let $S\in\{L_c,\mathcal U_\nu,K_e,M_\delta\}$, and let $P_S(x)$ denote the corresponding polynomial on the right-hand side of \eqref{eq:basic-vanishing-polynomials}. By definition, $P_S(\alpha)=0$ for every $\alpha\in S$, and hence
\[
 F_S(x)=\prod_{\alpha\in S}(x-\alpha)\mid P_S(x).
\]
The cardinalities established above imply $\deg P_S=|S|=\deg F_S$. Since $P_S$ and $F_S$ are both monic, $P_S=F_S$. The linear factors in the product are distinct, so all four polynomials have only simple roots.
\end{proof}

We now use these intersection formulas and vanishing polynomials to construct five families of Hermitian self-orthogonal GRS codes whose locator sets are complements of unions of selected subsets.
\subsection{The first family of Hermitian self-orthogonal GRS codes}

The first construction uses sets $\mathcal U_\nu$ with prescribed norm values and $\mathbb F_q^{*}$-cosets $M_\delta$.

Retain the notation $g$, $\rho$, and $\xi$ introduced above, and define
\[
 U_0=\{\rho^j:j\equiv0\pmod2\},
 \qquad
 U_1=\{\rho^j:j\equiv1\pmod2\}.
\]
Choose $\mathcal N_0$ and $\mathcal N_1$ to be subsets of the squares and nonsquares of $\mathbb F_q^{*}$, respectively, and write $|\mathcal N_i|=a_i\in\{0,1,\ldots,(q-1)/2\}$ for $i=0,1$. Choose arbitrary $D_i\subseteq U_i$ with $|D_i|=b_i\in\{0,1,\ldots,(q+1)/2\}$ for $i=0,1$. Put $a=a_0+a_1,  b=b_0+b_1\ge1,\ \kappa=a_0b_0+a_1b_1$, and assume $a+b\le q-1$. Define
\[
 B_U=\bigcup_{\nu\in\mathcal N_0\cup\mathcal N_1}\mathcal U_\nu,
 \qquad
 B_M=\bigcup_{\delta\in D_0\cup D_1}M_\delta,
 \qquad
 A=\mathbb F_{q^2}^{*}\setminus(B_U\cup B_M),
\]
and
\[
 N_1=q^2-1-a(q+1)-b(q-1)+2\kappa.
\]

\begin{theorem}\label{thm:first-grs}
For every integer $1\le k\le q-a-b$, there exists $\boldsymbol v=(v_\alpha)_{\alpha\in A}\in (\mathbb F_{q^2}^{*})^{N_1}$ such that $\GRS_k(A,\boldsymbol v)$ is a Hermitian self-orthogonal $[N_1,k,N_1-k+1]_{q^2}$ MDS code.
\end{theorem}

\begin{proof}
By Lemma~\ref{lem:coset-intersections}, the sets $\mathcal U_\nu$ for distinct norm values $\nu$ are pairwise disjoint and each has $q+1$ elements. Likewise, the sets $M_\delta$ for distinct parameters are pairwise disjoint and each has $q-1$ elements. The same lemma shows that, for $i,j\in\{0,1\}$, $\nu\in\mathcal N_i$, and $\delta\in D_j$, the sets $\mathcal U_\nu$ and $M_\delta$ intersect if and only if $i=j$, and each nonempty intersection has two elements. Distinct parameter pairs give disjoint intersections. Hence
\[
 |B_U\cap B_M|=2a_0b_0+2a_1b_1=2\kappa.
\]
The inclusion--exclusion principle yields $|A|=N_1$. Moreover, $a+b\le q-1$ implies $N_1\ge2b+2\kappa>0$.

Put $D=D_0\cup D_1$ and define
\[
\begin{aligned}
 P_U(x)=\prod_{\nu\in\mathcal N_0\cup\mathcal N_1}
 (x^{q+1}-\nu),
 \qquad
 P_M(x)=\prod_{\delta\in D}(x^{q-1}-\delta).
\end{aligned}
\]
By Lemma~\ref{lem:basic-vanishing-polynomials}, we have $P_U=F_{B_U}$ and $P_M=F_{B_M}$.

Set $P=P_UP_M$. In Lemma~\ref{lem:q-power-calculation}, take $H=P_U$ and $Q=P_M$. For every $\alpha\in A$, we have $\alpha\notin B_U$, so each factor $\alpha^{q+1}-\nu$ is a nonzero element of $\mathbb F_q$. Hence $H(\alpha)\in\mathbb F_q^{*}$.

Furthermore, $X=\mathbb F_{q^2}^{*}$, so $0\notin A$. Since $D_0\subseteq U_0$ and $D_1\subseteq U_1$, we have $D=D_0\cup D_1\subseteq U$ and $b\ge1$.
For every positive integer $k$ with $k\le q-a-b$,
\[
\begin{aligned}
 &q^2+q-1-bq-\deg H-(q+1)k\\
 &\qquad=(q+1)(q-a-b-k)+(b-1)\ge0.
\end{aligned}
\]
Thus the degree condition~\eqref{eq:specialized-degree} in Lemma~\ref{lem:q-power-calculation} holds, and the result follows from that lemma.
\end{proof}

\subsection{The second family of Hermitian self-orthogonal GRS codes}

In the second construction, the sets $\mathcal U_\nu$ in the first construction are replaced by sets $L_c$ with prescribed trace values.

For $c\in\mathbb F_q$ and $\delta\in U$, let $L_c$ and $M_\delta$ be as in Proposition~\ref{prop:intersection-formulas}. Choose $C=\{0\}\cup C_1$, where $C_1\subseteq\mathbb F_q^{*}$ and $|C_1|=u\ge0$, and choose a nonempty subset $D\subseteq U\setminus\{-1\}$ with $|D|=b\ge1$. Assume $u+b\le q-2$. Put
\[
 L_C=\bigcup_{c\in C}L_c,
 \qquad
 M_D=\bigcup_{\delta\in D}M_\delta,
 \qquad
 A=\mathbb F_{q^2}\setminus(L_C\cup M_D),
\]
and
\[
 N_2=q^2-(u+1)q-b(q-1)+ub.
\]

\begin{theorem}\label{thm:second-grs}
For every integer $1\le k\le q-u-b-1$, there exists $\boldsymbol v=(v_\alpha)_{\alpha\in A}\in (\mathbb F_{q^2}^{*})^{N_2}$ such that $\GRS_k(A,\boldsymbol v)$ is a Hermitian self-orthogonal $[N_2,k,N_2-k+1]_{q^2}$ MDS code.
\end{theorem}

\begin{proof}
By Proposition~\ref{prop:intersection-formulas}, each $M_\delta$ meets each $L_c$ with $c\ne0$ in exactly one point and is disjoint from $L_0$. Distinct parameter pairs $(c,\delta)$ give distinct intersection points. Hence $|L_C\cap M_D|=ub$. The inclusion--exclusion principle yields $|A|=N_2=(q-u-1)(q-b)>0$, since both factors are positive by $u+b\le q-2$.

Define
\[
 P_C(x)=\prod_{c\in C}(x^q+x-c),
 \qquad
 P_M(x)=\prod_{\delta\in D}(x^{q-1}-\delta).
\]
By Lemma~\ref{lem:basic-vanishing-polynomials}, we have $P_C=F_{L_C}$ and $P_M=F_{M_D}$.

Since $0\in L_0$, we have $0\notin A$. In Lemma~\ref{lem:q-power-calculation}, take $X=\mathbb F_{q^2}$, $H=P_C$, and $Q=P_M$. Then $P=P_CP_M=HQ$, $c_0=1$, $b\ge1$, and $D\subseteq U$. For every $\alpha\in A$, each factor $\Tr(\alpha)-c$ belongs to $\mathbb F_q^{*}$, so $H(\alpha)\in\mathbb F_q^{*}$. Moreover, $\deg H=(u+1)q$.

For every positive integer $k$ with $k\le q-u-b-1$,
\[
\begin{aligned}
 &q^2+q-1-bq-\deg H-(q+1)k\\
 &\qquad=(q+1)(q-u-b-1-k)+(u+b)\ge0.
\end{aligned}
\]
The result therefore follows from Lemma~\ref{lem:q-power-calculation}.
\end{proof}

\subsection{The third family of Hermitian self-orthogonal GRS codes}

The third construction combines sets $L_c$ with prescribed trace values and sets $\mathcal U_\nu$ with prescribed norm values.

Choose nonempty subsets $C\subseteq\mathbb F_q$ and $\mathcal N\subseteq\mathbb F_q^{*}$ such that $|C|=R\ge1$, $|\mathcal N|=a\ge1$, and $R+a\le q-1$. Put
\[
 L_C=\bigcup_{c\in C}L_c,
 \qquad
 \mathcal U_{\mathcal N}=\bigcup_{\nu\in\mathcal N}\mathcal U_\nu,
 \qquad
 A=\mathbb F_{q^2}\setminus(L_C\cup\mathcal U_{\mathcal N}),
\]
and write
\[
 \sigma=|L_C\cap\mathcal U_{\mathcal N}|
 =\sum_{c\in C,\,\nu\in\mathcal N}
 \bigl(1-\chi(c^2-4\nu)\bigr),
\]
\[
 N_3=q^2-Rq-a(q+1)+\sigma.
\]

\begin{theorem}\label{thm:third-grs}
For every integer $1\le k\le q-R-a$, there exists $\boldsymbol v=(v_\alpha)_{\alpha\in A}\in (\mathbb F_{q^2}^{*})^{N_3}$ such that $\GRS_k(A,\boldsymbol v)$ is a Hermitian self-orthogonal $[N_3,k,N_3-k+1]_{q^2}$ MDS code.
\end{theorem}

\begin{proof}
The sets $L_c$ for distinct trace values $c$ are pairwise disjoint, as are the sets $\mathcal U_\nu$ for distinct norm values $\nu$. Lemma~\ref{lem:trace-norm-intersections} therefore gives the stated expression for $\sigma$. The inclusion--exclusion principle then yields $|A|=N_3$. Since $R+a\le q-1$, we have $N_3\ge q-a>0$.

Define
\[
 P_C(x)=\prod_{c\in C}(x^q+x-c),
 \qquad
 P_{\mathcal N}(x)=\prod_{\nu\in\mathcal N}(x^{q+1}-\nu).
\]
By Lemma~\ref{lem:basic-vanishing-polynomials}, we have $P_C=F_{L_C}$ and $P_{\mathcal N}=F_{\mathcal U_{\mathcal N}}$. Put $P=P_CP_{\mathcal N}$, so $\deg P=Rq+a(q+1)$. For every $\alpha\in A$, we have $P_C(\alpha),\ P_{\mathcal N}(\alpha)\in\mathbb F_q^{*}$, and hence $P(\alpha)^q=P(\alpha)$. Thus the hypotheses of Theorem~\ref{thm:general-construction} hold with $X=\mathbb F_{q^2}$, $\varepsilon_X=0$, $m=0$, $\gamma=1$, and $c_0=1$.

For every positive integer $k$ with $k\le q-R-a$,
\[
\begin{aligned}
 &q^2+q-1-\deg P-(q+1)k\\
 &\qquad=(q+1)(q-R-a-k)+(R-1)\ge0.
\end{aligned}
\]
Thus the degree condition~\eqref{eq:general-degree} in Theorem~\ref{thm:general-construction} holds, proving the result.
\end{proof}

\subsection{The fourth family of Hermitian self-orthogonal GRS codes}

The fourth construction combines the two families of sets $L_c$ and $K_e$ with prescribed trace values with the $\mathbb F_q^{*}$-cosets $M_\delta$.

Choose $\omega\in\mathbb F_{q^2}^{*}$ with $\omega^q=-\omega$. Let $C=\{0\}\cup C_1$, where $C_1\subseteq\mathbb F_q^{*}$ and $|C_1|=u\ge1$, and choose $E\subseteq\mathbb F_q^{*}$ with $|E|=v\ge1$. Choose a nonempty subset $D\subseteq U\setminus\{\pm1\}$ and write $|D|=b\ge1$. Assume $u+v+b\le q-2$. Put
\[
 L_C=\bigcup_{c\in C}L_c,
 \quad
 K_E=\bigcup_{e\in E}K_e,
 \quad
 M_D=\bigcup_{\delta\in D}M_\delta,
 \quad
 A=\mathbb F_{q^2}\setminus(L_C\cup K_E\cup M_D),
\]
and define
\[
 \tau_D
 =\#\left\{(c,e)\in C_1\times E:
 \left(\frac12\left(c-\frac e\omega\right)\right)^{q-1}\in D\right\},
\]
\[
 \begin{aligned}
 N_4={}&q^2-(u+1)q-vq-b(q-1)\\
       &+(u+1)v+b(u+v)-\tau_D.
\end{aligned}
\]

\begin{theorem}\label{thm:fourth-grs}
For every integer $1\le k\le q-u-v-b-1$, there exists $\boldsymbol v=(v_\alpha)_{\alpha\in A}\in (\mathbb F_{q^2}^{*})^{N_4}$ such that $\GRS_k(A,\boldsymbol v)$ is a Hermitian self-orthogonal $[N_4,k,N_4-k+1]_{q^2}$ MDS code.
\end{theorem}

\begin{proof}
By Proposition~\ref{prop:intersection-formulas},
\[
 |L_C\cap K_E|=(u+1)v,
 \qquad
 |L_C\cap M_D|=ub,
 \qquad
 |K_E\cap M_D|=vb.
\]
The triple intersection has size $\tau_D$. Hence the inclusion--exclusion principle yields $|A|=N_4$. Since $u+v+b\le q-2$, the sum of the cardinalities of the three deleted sets is at most $q^2-q-b$. Thus $N_4\ge q+b>0$.

Define
\[
 P_C(x)=\prod_{c\in C}(x^q+x-c),
 \qquad
 P_E(x)=\prod_{e\in E}(\omega(x^q-x)-e),
\]
\[
 P_M(x)=\prod_{\delta\in D}(x^{q-1}-\delta).
\]
By Lemma~\ref{lem:basic-vanishing-polynomials}, we have $P_C=F_{L_C}$, $P_E=\omega^vF_{K_E}$, and $P_M=F_{M_D}$.

Since $0\in L_0$, we have $0\notin A$. In Lemma~\ref{lem:q-power-calculation}, take $X=\mathbb F_{q^2}$, $H=P_CP_E$, and $Q=P_M$. Then $P=P_CP_EP_M=HQ$, $c_0=\omega^v$, $b\ge1$, and $D\subseteq U$. For every $\alpha\in A$, the factors of the first two polynomials have the forms $\Tr(\alpha)-c$ and $\Tr(-\omega\alpha)-e$, respectively. Since $\alpha\notin L_C\cup K_E$, all these factors belong to $\mathbb F_q^{*}$, and hence $H(\alpha)\in\mathbb F_q^{*}$. Comparing degrees gives $\deg H=(u+v+1)q$.

For every positive integer $k$ with $k\le q-u-v-b-1$,
\[
\begin{aligned}
 &q^2+q-1-bq-\deg H-(q+1)k\\
 &\qquad=(q+1)(q-u-v-b-1-k)+(u+v+b)\ge0.
\end{aligned}
\]
The result therefore follows from Lemma~\ref{lem:q-power-calculation}.
\end{proof}

\subsection{The fifth family of Hermitian self-orthogonal GRS codes}

The fifth construction combines the two families of sets $L_c$ and $K_e$ with prescribed trace values with the sets $\mathcal U_\nu$ having prescribed norm values.

Choose nonempty subsets $C,E\subseteq\mathbb F_q$ and $\mathcal N\subseteq\mathbb F_q^{*}$ such that $|C|=R\ge1$, $|E|=V\ge1$, $|\mathcal N|=a\ge1$, and $R+V+a\le q-1$. Put
\[
 L_C=\bigcup_{c\in C}L_c,
 \quad
 K_E=\bigcup_{e\in E}K_e,
 \quad
 \mathcal U_{\mathcal N}=\bigcup_{\nu\in\mathcal N}\mathcal U_\nu,
 \quad
 A=\mathbb F_{q^2}\setminus(L_C\cup K_E\cup\mathcal U_{\mathcal N}).
\]
Write
\[
\begin{aligned}
 \sigma_+&=|L_C\cap\mathcal U_{\mathcal N}|
 =\sum_{c\in C,\,\nu\in\mathcal N}
 \bigl(1-\chi(c^2-4\nu)\bigr),\\
 \sigma_-&=|K_E\cap\mathcal U_{\mathcal N}|
 =\sum_{e\in E,\,\nu\in\mathcal N}
 \left(1+\chi\left(\nu+\frac{e^2}{4\Delta}\right)\right),
\end{aligned}
\]
\[
 \tau_N=\#\left\{(c,e)\in C\times E:
 \frac{c^2-e^2/\Delta}{4}\in\mathcal N\right\},
\]
and
\[
 \begin{aligned}
 N_5={}&q^2-Rq-Vq-a(q+1)\\
       &+RV+\sigma_++\sigma_- -\tau_N.
\end{aligned}
\]

\begin{theorem}\label{thm:fifth-grs}
For every integer $1\le k\le q-R-V-a$, there exists $\boldsymbol v=(v_\alpha)_{\alpha\in A}\in (\mathbb F_{q^2}^{*})^{N_5}$ such that $\GRS_k(A,\boldsymbol v)$ is a Hermitian self-orthogonal $[N_5,k,N_5-k+1]_{q^2}$ MDS code.
\end{theorem}

\begin{proof}
Within each of the three families, sets corresponding to distinct parameters are pairwise disjoint. Lemma~\ref{lem:trace-norm-intersections} therefore gives the stated expressions for $\sigma_+$ and $\sigma_-$. Proposition~\ref{prop:intersection-formulas} also gives
\[
 |L_C\cap K_E|=RV.
\]
The triple intersection has size $\tau_N$. Hence the inclusion--exclusion principle yields $|A|=N_5$. Since $R+V+a\le q-1$, the sum of the cardinalities of the three deleted sets is at most $q(q-1)+a$. Thus $N_5\ge q-a>0$.

Put
\[
 P_C(x)=\prod_{c\in C}(x^q+x-c),
 \quad
 P_E(x)=\prod_{e\in E}(\omega(x^q-x)-e),
\]
\[
 P_{\mathcal N}(x)=\prod_{\nu\in\mathcal N}(x^{q+1}-\nu).
\]
By Lemma~\ref{lem:basic-vanishing-polynomials}, we have $P_C=F_{L_C}$, $P_E=\omega^VF_{K_E}$, and $P_{\mathcal N}=F_{\mathcal U_{\mathcal N}}$. Set $P=P_CP_EP_{\mathcal N}$. Its leading coefficient is $\omega^V$, and $\deg P=(R+V)q+a(q+1)$. For every $\alpha\in A$, all three factors take values in $\mathbb F_q^{*}$, so $P(\alpha)^q=P(\alpha)$. Thus the hypotheses of Theorem~\ref{thm:general-construction} hold with $X=\mathbb F_{q^2}$, $\varepsilon_X=0$, $m=0$, $\gamma=1$, and $c_0=\omega^V$.

For every positive integer $k$ with $k\le q-R-V-a$,
\[
\begin{aligned}
 &q^2+q-1-\deg P-(q+1)k\\
 &\qquad=(q+1)(q-R-V-a-k)+(R+V-1)\ge0.
\end{aligned}
\]
Thus the degree condition~\eqref{eq:general-degree} in Theorem~\ref{thm:general-construction} holds, proving the result.
\end{proof}

\section{Quantum MDS codes and parameter comparisons}\label{sec:quantum}\label{sec:comparison}

We now apply the five families of Hermitian self-orthogonal GRS codes in Section~\ref{sec:classical} to construct quantum MDS codes and express their lengths and minimum-distance ranges in a common form. We then compare their parameters with those of codes obtainable from previously known constructions, either directly or via the propagation rule. In particular, we determine the minimum distances attainable at a fixed length by constructions based on trace maps, linear transformations, and cosets of multiplicative subgroups, and give sufficient conditions for strict improvement together with exact gains or lower bounds on the gains. We begin by recalling the Hermitian construction.

\begin{theorem}[Hermitian construction \cite{AshikhminKnill2001,KetkarEtAl2006}]\label{thm:hermitian}
If an \mbox{$[n,k,n-k+1]_{q^2}$} MDS code $C$ satisfies $C\subseteq C^{\perp_H}$, then there exists a pure $q$-ary quantum MDS code with parameters $\qcode{n}{n-2k}{k+1}$.
\end{theorem}

Here $C^{\perp_H}$ is an $[n,n-k,k+1]_{q^2}$ MDS code, whereas $C$ has minimum distance $n-k+1\ge k+1$. Thus the resulting quantum code is pure to distance $k+1$, and the propagation rule for pure quantum codes applies.

\subsection{Parameters of the five families of quantum MDS codes}\label{subsec:quantum-families}

Applying the Hermitian construction to the five families in Section~\ref{sec:classical} gives the following families of quantum MDS codes. For the comparisons below, we express their lengths in a common form and retain the notation for the sets, cardinalities, and intersection parameters from the preceding section.

\begin{corollary}\label{cor:quantum-families}
Under the respective parameter conditions for the five constructions, let $S,\theta$ be as in Table~\ref{tab:general-comparison-parameters} and put $n=q^2-Sq+\theta$. Then, for every integer $2\le d\le q-S+1$, there exists a quantum MDS code with parameters $\qcode{n}{n-2d+2}{d}$.
\end{corollary}

\begin{proof}
Take $k=d-1$ in the corresponding theorem for classical codes and apply Theorem~\ref{thm:hermitian}.
\end{proof}

\begin{table}[htbp]
\centering\small
\caption{Parameters of the five families: $n=q^2-Sq+\theta$ and $2\le d\le q-S+1$}
\label{tab:general-comparison-parameters}
\setlength{\tabcolsep}{5pt}
\begin{tabular}{cll}
\toprule
Family & $S$ & $\theta$\\
\midrule
I & $a+b$ & $-1-a+b+2\kappa$\\
II & $u+b+1$ & $(u+1)b$\\
III & $R+a$ & $\sigma-a$\\
IV & $u+v+b+1$ & $(u+1)v+b(u+v+1)-\tau_D$\\
V & $R+V+a$ & $RV+\sigma_++\sigma_- -\tau_N-a$\\
\bottomrule
\end{tabular}
\end{table}

To compare the minimum distances of these families, write
\begin{equation}\label{eq:comparison-parameters}
 n=q^2-Sq+\theta,\qquad d_*=q-S+1.
\end{equation}
Here $d_*$ is the upper endpoint of the minimum-distance range in Corollary~\ref{cor:quantum-families}, and this endpoint is attained in each of our constructions.

This parameterization also describes the effect of intersections on the code length. Once the cardinalities determining $S$ are fixed, the lengths of the first, third, fourth, and fifth families still depend on the corresponding intersection sizes. Among admissible choices, varying these sizes can change the length without changing $d_*$. The length of the second family is determined by $u$ and $b$. When the numbers of selected trace values, norm values, and cosets are fixed, the resulting families have lengths $n=q^2-O(q)$ and contain codes with minimum distance $d_*=q-O(1)$. Moreover, when $S\le(q-1)/2$, these families contain quantum MDS codes with minimum distances greater than $q/2+1$.

The comparisons below include both the codes obtained directly from previously known constructions and all codes derived from them via the propagation rule; the relevant families are collected in \ref{app:known}. The propagation rule \cite[Lemma~1.4]{FangFu2019} states that if there exists a pure quantum MDS code with parameters $\qcode{n_0}{n_0-2d_0+2}{d_0}$, then, for every integer $0\le\ell\le d_0-2$, there exists a quantum MDS code with parameters
\[
 \qcode{n_0-\ell}{n_0-2d_0+\ell+2}{d_0-\ell}.
\]
Thus, even if a previously known family has a different original length, one must determine whether the propagation rule yields a code with the target parameters.

\subsection{Comparison with constructions based on trace maps}\label{subsec:comparisons}

Fang, Wen, and Fu \cite[Theorems~5--8]{FangWenFu2024} constructed two families of quantum MDS codes. In the first, trace maps are combined with finite-field subspaces; in the second, trace maps are combined with norm maps. Their respective lengths are $(s+t)q-st$ and $1+sq+t(q+1)-N_{s,t}$, and in both cases the minimum distance is at most $\min\{s,t\}+1$. These families are listed as R3 and R4 in the appendix. To compare them directly with our codes, we derive upper bounds, depending only on $q$ and the target length $n$, for all codes obtained from these families either directly or via the propagation rule.

\begin{proposition}\label{prop:fwf}
Let $1\le n<q^2$ and write $D=q^2-n$. Every code of length $n$ obtained from the trace--subspace construction above, either directly or via the propagation rule, has minimum distance satisfying
\begin{equation}\label{eq:trace-subspace-bound}
 d\le q+1-\lceil\sqrt D\rceil.
\end{equation}
Every code of length $n$ obtained from the trace--norm construction above, either directly or via the propagation rule, has minimum distance satisfying
\begin{equation}\label{eq:trace-norm-bound}
 d\le q+1-\left\lceil\frac{1+\sqrt{1+2D}}2\right\rceil.
\end{equation}
\end{proposition}

\begin{proof}
Let $n_0$ and $d_0$ denote the length and minimum distance, respectively, of an original code, and put $m=q-d_0+1$. For the first construction, $d_0\le\min\{s,t\}+1$ implies $q-s\le m$ and $q-t\le m$, so
\[
 q^2-n_0=(q-s)(q-t)\le m^2.
\]
If the length reduction is $\ell$, then $n=n_0-\ell$ and $d=d_0-\ell$. Hence
\[
 q^2-n\le m^2+\ell\le(m+\ell)^2=(q-d+1)^2.
\]
Here $m,\ell$ are nonnegative integers. The second inequality is immediate when $\ell=0$; when $\ell\ge1$, it follows from $2m\ell+\ell^2\ge\ell$. Rearranging and taking ceilings gives \eqref{eq:trace-subspace-bound}.

For the second construction, the locator set contains zero. Each element of its complement corresponds to an unselected trace value and an unselected nonzero norm value. There are $(q-s)(q-1-t)$ such pairs, each corresponding to at most two elements. Since $d_0\le\min\{s,t\}+1$ and $t\le q-1$, we have $m\ge1$, and hence
\[
 q^2-n_0\le2(q-s)(q-1-t)\le2m(m-1).
\]
After reducing the length by $\ell$ via the propagation rule,
\[
 q^2-n\le2m(m-1)+\ell\le2(m+\ell)(m+\ell-1),
\]
since the difference between the rightmost and middle expressions is $\ell(4m+2\ell-3)\ge0$. Substituting $m+\ell=q-d+1$ and solving the quadratic inequality gives \eqref{eq:trace-norm-bound}.
\end{proof}

Comparing these two upper bounds with our value of $d_*$ gives a sufficient condition that applies to all five constructions.

\begin{corollary}\label{cor:general-trace-improvement}
For any admissible parameter choice in Table~\ref{tab:general-comparison-parameters}, if $S\ge2$ and
\begin{equation}\label{eq:general-trace-improvement}
 Sq-\theta>2S(S-1),
\end{equation}
then the minimum distance $d_*=q-S+1$ attained by our construction at length $q^2-Sq+\theta$ strictly exceeds the corresponding upper bounds for every code of the same length obtained from R3 or R4, either directly or via the propagation rule. For such codes, the respective lower bounds on the gain in minimum distance are
\[
 \lceil\sqrt{Sq-\theta}\rceil-S,
 \qquad
 \left\lceil\frac{1+\sqrt{1+2(Sq-\theta)}}2\right\rceil-S.
\]
\end{corollary}

\begin{proof}
Since $S\ge2$, we have $2S(S-1)\ge S^2$. Condition~\eqref{eq:general-trace-improvement} therefore ensures that both ceiling terms in Proposition~\ref{prop:fwf} are at least $S+1$. Every code of the length under consideration obtained from either family has minimum distance at most $q-S<d_*$. Subtracting the respective upper bounds from $d_*$ gives the stated lower bounds on the gains in minimum distance.
\end{proof}

The corollary applies to arbitrary admissible parameters in all five constructions. Fix the numbers of selected trace values, norm values, and cosets, and suppose $S\ge2$. Then $S$ is constant, while the intersection sizes and $\theta$ are bounded above by quantities depending only on these numbers. Consequently, for all sufficiently large $q$, condition~\eqref{eq:general-trace-improvement} holds for every admissible parameter choice with these cardinalities. Thus strict improvements in minimum distance at the same length occur for infinitely many values of $q$.

\subsection{Comparison with constructions based on linear transformations}

Li, Liu, and Jiang \cite[Theorems~4.4--4.6 and Corollary~4.9]{LiLiuJiang2026} constructed three families of quantum MDS codes with original lengths $2qr$, $(2r-1)q+1$, and $(q+1)r$, respectively, and common minimum distance $D(r)=(q+1)/2+r$, where $1\le r\le(q-1)/2$. Because the parameter $r$ determines both the original length and the minimum distance, fixing a target length $n$ allows us to determine the largest minimum distance attainable by these constructions, including codes derived from them via the propagation rule. For target lengths greater than $(q^2-1)/2$, neither the third family nor any code derived from it via the propagation rule can attain the target length, so only the first two families need to be considered.

\begin{proposition}\label{prop:llj}
Let $n$ be an integer satisfying $(q^2-1)/2<n<q^2$, and put
\[
 r_1=\left\lceil\frac n{2q}\right\rceil,
 \qquad r_2=\left\lceil\frac{n+q-1}{2q}\right\rceil,
 \qquad L_1=2qr_1,\quad L_2=(2r_2-1)q+1.
\]
For each $j\in\{1,2\}$ such that $1\le r_j\le(q-1)/2$, define
\[
 d_j=\frac{q+1}{2}+r_j-(L_j-n).
\]
If at least one $d_j$ is at least $2$, then the maximum of the values $d_j\ge2$ is the largest minimum distance attainable at length $n$ by the first two constructions above, including all codes derived from them via the propagation rule. If no such $d_j$ exists, these constructions yield no quantum MDS code of length $n$ with minimum distance at least $2$, either directly or via the propagation rule.
\end{proposition}

\begin{proof}
The original lengths in the third construction are at most $(q+1)(q-1)/2$, and the propagation rule does not increase the length. Thus only the first two families need to be considered in the stated range. For family $j$, the integer $r_j$ is the smallest parameter for which the original length is at least $n$. If $r_j>(q-1)/2$, all original lengths in this family are less than $n$, and the propagation rule cannot produce the required length.

Now assume $1\le r_j\le(q-1)/2$. For any integer $r_j\le r\le(q-1)/2$, the original code in this family has length $L_j+2q(r-r_j)$ and minimum distance $(q+1)/2+r$. Applying the propagation rule to reduce the length to $n$ gives minimum distance
\[
 \frac{q+1}{2}+r-\bigl(L_j+2q(r-r_j)-n\bigr)
 =d_j-(2q-1)(r-r_j).
\]

Since $2q-1>0$, this expression is maximized at $r=r_j$, where it equals $d_j$. Therefore, if $d_j\ge2$, the choice $r=r_j$ gives the largest minimum distance attainable at length $n$ from this family, whether directly or via the propagation rule. If $d_j<2$, no admissible value of $r$ yields a code of length $n$ with minimum distance at least $2$. Taking the largest candidate satisfying the parameter conditions and $d_j\ge2$ proves the result.
\end{proof}

Substituting our length $n=q^2-Sq+\theta$ into Proposition~\ref{prop:llj} gives the following comparison formulas.

\begin{corollary}\label{cor:llj-comparison}
Let $S,\theta$ be integers satisfying $2\le S\le q-2$ and $-q+2\le\theta\le q$, and suppose $n=q^2-Sq+\theta>(q^2-1)/2$. Put $\beta_S=0$ if $S$ is odd and $\beta_S=1$ if $S$ is even. If the right-hand side below is at least $2$, then the largest minimum distance attainable at length $n$ by the first two constructions in \cite{LiLiuJiang2026}, including codes derived from them via the propagation rule, is
\begin{equation}\label{eq:llj-band}
 d_{\mathrm{LLJ}}=
 \begin{cases}
 q+1-\lceil S/2\rceil-\beta_S+\theta,&\theta\le\beta_S,\\
 \theta+1-\lceil S/2\rceil,&\theta>\beta_S.
 \end{cases}
\end{equation}
If the right-hand side is less than $2$, these constructions yield no quantum MDS code of length $n$ with minimum distance at least $2$, either directly or via the propagation rule. If $(S,\theta)$ is attained by one of our constructions and $d_{\mathrm{LLJ}}\ge2$, then the difference between our minimum distance $d_*=q-S+1$ and this largest minimum distance is
\begin{equation}\label{eq:general-llj-gain}
 d_*-d_{\mathrm{LLJ}}=
 \begin{cases}
 \beta_S-\lfloor S/2\rfloor-\theta,&\theta\le\beta_S,\\
 q-\lfloor S/2\rfloor-\theta,&\theta>\beta_S.
 \end{cases}
\end{equation}
\end{corollary}

\begin{proof}
Arrange the original codes from the first two constructions in increasing order of length. Consecutive lengths differ by $q-1$ or $q+1$, while the corresponding minimum distance increases by at most $1$. Thus, after applying the propagation rule to obtain a fixed length $n$, a larger original length gives a smaller minimum distance. It is therefore enough to choose the shortest original code whose length is at least $n$ and then verify that the required reduction is allowed by the propagation rule.

First suppose $\theta\le\beta_S$. Since $q$ is odd, when $S$ is odd we take $r=(q-S)/2$ in the first construction, giving original length $2qr=q^2-Sq$. When $S$ is even, we take $r=(q-S+1)/2$ in the second construction, giving original length $(2r-1)q+1=q^2-Sq+1$. The inequalities $2\le S\le q-2$ ensure that $1\le r\le(q-1)/2$ in both cases. Thus the selected original length and minimum distance can be written uniformly as
\[
 n_0=q^2-Sq+\beta_S,\qquad
 d_0=\frac{q+1}{2}+r=q+1-\lceil S/2\rceil.
\]

In the sequence of original lengths from the first two constructions, the length immediately preceding $n_0$ is $n_{\mathrm{prev}}=q^2-(S+1)q+1-\beta_S$. Since $-q+2\le\theta\le\beta_S$,
\[
 \begin{aligned}
 n-n_{\mathrm{prev}}&=q+\theta-1+\beta_S\ge1+\beta_S>0,\\
 n_0-n&=\beta_S-\theta\ge0.
 \end{aligned}
\]
Hence $n_{\mathrm{prev}}<n\le n_0$, so $n_0$ is indeed the smallest original length not less than $n$. Applying the propagation rule to reduce the length to $n$ gives minimum distance $d_0-(n_0-n)=q+1-\lceil S/2\rceil-\beta_S+\theta$.

Similarly, if $\theta>\beta_S$, then $q^2-Sq+\beta_S<n$, so we choose the next original code, with parameters
\[
 n_0=q^2-(S-1)q+1-\beta_S,\qquad
 d_0=q+1-\lceil(S-1)/2\rceil.
\]
The corresponding parameter $r=(q-S+2-\beta_S)/2$ satisfies $1\le r\le(q-1)/2$. Moreover, $\theta\le q$ gives
\[
 n_0-n=q+1-\beta_S-\theta\ge1-\beta_S\ge0.
\]
Thus $n_0$ is again the smallest original length not less than $n$, and applying the propagation rule gives minimum distance $d_0-(n_0-n)=\theta+1-\lceil S/2\rceil$. By Proposition~\ref{prop:llj}, the computed value is the largest attainable minimum distance if it is at least $2$; otherwise these constructions yield no code of the required length with minimum distance at least $2$. This proves \eqref{eq:llj-band}. Subtracting $d_{\mathrm{LLJ}}$ from $d_*=q-S+1$ in the two cases and using $S-\lceil S/2\rceil=\lfloor S/2\rfloor$ gives \eqref{eq:general-llj-gain}.
\end{proof}

Applying these comparison formulas to our second construction gives the following corollary.

\begin{corollary}\label{cor:second-comparison}
Let $R,b$ be positive integers, put $S=R+b\ge4$, and suppose $q>2S$. Define
\[
 h=\left\lfloor\frac{Rb}{q}\right\rfloor,
 \qquad T=S-h,\qquad t=Rb-hq.
\]
If
\[
 \lceil T/2\rceil+1\le t<q-S+\lceil T/2\rceil,
\]
then, for every integer $2\le d\le d_*=q-S+1$, there exists a quantum MDS code $\qcode{n}{n-2d+2}{d}$ of length $n=(q-R)(q-b)$. The attained minimum distance $d_*$ is greater than $q/2+1$ and strictly exceeds the largest minimum distance attainable at this length from the linear-transformation constructions of \cite{LiLiuJiang2026}, including codes derived from them via the propagation rule. This largest minimum distance and the corresponding gain are, respectively,
\[
 d_{\mathrm{LLJ}}=t+1-\lceil T/2\rceil,
 \qquad d_*-d_{\mathrm{LLJ}}=q-S+\lceil T/2\rceil-t>0.
\]
In particular, if
\begin{equation}\label{eq:second-two-parameter}
 q>\max\left\{2S,\ Rb+\lfloor S/2\rfloor\right\},
\end{equation}
then all the conditions above hold, and
\[
 d_{\mathrm{LLJ}}=Rb+1-\lceil S/2\rceil,
 \qquad d_*-d_{\mathrm{LLJ}}=q-\lfloor S/2\rfloor-Rb>0.
\]
\end{corollary}

\begin{proof}
Take $u=R-1\ge0$ in the second construction. Since $q>2S$, we have $R,b<q$ and $u+b=S-1\le q-2$, so the required sets can be chosen. Corollary~\ref{cor:quantum-families} gives the stated quantum MDS codes and minimum-distance range. Moreover, $q>2S$ implies $d_*=q-S+1>q/2+1$, and
\[
 \frac{q^2-1}{2}<\frac{q^2}{2}<q^2-Sq+Rb=n<q^2.
\]

Since $R,b<q$, we have $0\le h<\min\{R,b\}$, and hence $2\le T\le S\le(q-1)/2\le q-2$. As $0\le t\le q-1$, we may write the length as
\[
 n=q^2-Sq+Rb=q^2-Tq+t,
\]
and apply Corollary~\ref{cor:llj-comparison} to the linear-transformation constructions. The condition $t\ge\lceil T/2\rceil+1$ ensures $t>\beta_T$ and that the value in the second branch is at least $2$. Hence the largest minimum distance attainable at length $n$ from the linear-transformation constructions, including codes derived from them via the propagation rule, is $d_{\mathrm{LLJ}}=t+1-\lceil T/2\rceil$. The minimum distance attained by our second construction is still $d_*=q-S+1$. Subtracting gives the stated gain, which is positive by $t<q-S+\lceil T/2\rceil$.

Finally, if \eqref{eq:second-two-parameter} holds, then $Rb<q$, so $h=0$, $T=S$, and $t=Rb$. Since $R,b\ge1$ and $S\ge4$,
\[
 Rb\ge S-1\ge\lceil S/2\rceil+1,
 \qquad
 Rb<q-\lfloor S/2\rfloor=q-S+\lceil S/2\rceil.
\]
This verifies the conditions on $t$. Substituting $T=S$ and $t=Rb$ into the formulas for $d_{\mathrm{LLJ}}$ and the gain gives the final two expressions. Once positive integers $R,b$ with $R+b\ge4$ are fixed, condition~\eqref{eq:second-two-parameter} holds for all sufficiently large $q$. Thus we obtain a strict improvement in minimum distance at the same length for infinitely many values of $q$. In this case, the gain $q-\lfloor S/2\rfloor-Rb$ grows linearly with $q$.
\end{proof}

To obtain more explicit parameter formulas, we now fix the cardinalities of the parameter sets in the five constructions and obtain the five infinite subfamilies listed in Table~\ref{tab:infinite-comparison}.
\begin{corollary}\label{cor:infinite-comparison}
For each row of Table~\ref{tab:infinite-comparison}, whenever $q$ satisfies the stated condition, there exists a quantum MDS code with parameters $\qcode{n}{n-2d+2}{d}$ for every integer $2\le d\le d_*$. In the table, $d_{\mathrm{LLJ}}$ denotes the largest minimum distance attainable at the same length by the first two constructions in \cite[Corollary~4.9]{LiLiuJiang2026}, including codes derived from them via the propagation rule. The last column records the gain $d_*-d_{\mathrm{LLJ}}$.
\end{corollary}

\begin{table}[htbp]
\centering\small
\caption{Five infinite subfamilies and the best minimum distances attainable at the same length from previously known linear-transformation constructions}
\label{tab:infinite-comparison}
\setlength{\tabcolsep}{4pt}
\begin{tabular}{clcccc}
\toprule
Family&$n$&Range of $q$&$d_*$&$d_{\mathrm{LLJ}}$&Gain\\
\midrule
I&$q^2-4q-3$&$q\ge9$&$q-3$&$q-5$&$2$\\
II&$q^2-5q+6$&$q\ge11$&$q-4$&$4$&$q-8$\\
III&$q^2-3q-2$&$q\ge7$&$q-2$&$q-3$&$1$\\
IV&$q^2-5q+7$&$q\ge11$&$q-4$&$5$&$q-9$\\
V&$q^2-4q+5$&$q\ge9$&$q-3$&$4$&$q-7$\\
\bottomrule
\end{tabular}
\end{table}

\begin{proof}
Let $\rho$ be a generator of $U$, and choose $\omega\in\mathbb F_{q^2}^{*}$ such that $\omega^q=-\omega$. Write $\Delta=\omega^2$ and $\delta=(\omega+1)/(\omega-1)$. We specify the sets required for the five constructions in turn.

For the first construction, take $\mathcal N_0=\varnothing$, $|\mathcal N_1|=3$, $D_0=\{\rho^2\}$, and $D_1=\varnothing$. The corresponding parameters are $(a,b,\kappa)=(3,1,0)$.

For the second construction, take $C=\{0,1,2\}$ and $D=\{\rho,\rho^2\}$. The corresponding parameters are $u=b=2$.

For the third construction, choose two distinct squares $\Delta_1,\Delta_2$ in $\mathbb F_q^{*}$, both different from $1$, and take $C=\{1\}$ and $\mathcal N=\{(1-\Delta_1)/4,(1-\Delta_2)/4\}$. Such squares exist because $q\ge7$, and the resulting two norm values are nonzero and distinct. The corresponding discriminants are $\Delta_1,\Delta_2$, respectively, so $\sigma=0$.

For the fourth construction, take $C=\{0,1\}$, $E=\{1,-1\}$, and $D=\{\delta\}$. Since $\delta^q=\delta^{-1}$ and $\delta\ne\pm1$, we have $D\subseteq U\setminus\{\pm1\}$. To compute the number $\tau_D$ of triple intersections, it suffices to consider the two points in $L_1\cap K_1$ and $L_1\cap K_{-1}$. They satisfy $x_{1,1}^{q-1}=\delta$ and $x_{1,-1}^{q-1}=\delta^{-1}\notin D$, respectively. Thus only the former belongs to $M_\delta$, and hence $(u,v,b,\tau_D)=(1,2,1,1)$.

For the fifth construction, take $C=\{1\}$, $E=\{1,-1\}$, and $\mathcal N=\{(1-\Delta^{-1})/4\}$. The chosen norm value is nonzero. The discriminant for the intersection of $L_1$ with the coset of $U$ corresponding to this norm value is the nonsquare $\Delta^{-1}$, so $\sigma_+=2$. For $e=\pm1$, the quadratic equations for the intersections of $K_e$ with this coset both have right-hand side $1/4$, so $\sigma_-=4$. Moreover, the norms of both $x_{1,1}$ and $x_{1,-1}$ belong to $\mathcal N$, giving $\tau_N=2$. Consequently, $(R,V,a)=(1,2,1)$ and $(\sigma_+,\sigma_-,\tau_N)=(2,4,2)$.

For the ranges of $q$ listed in the table, all these sets satisfy the hypotheses of the corresponding classical-code construction theorems. Corollary~\ref{cor:quantum-families} gives the length $n$ and the minimum-distance range $2\le d\le d_*$ in each row. The corresponding pairs $(S,\theta)$ are $(4,-3)$, $(5,6)$, $(3,-2)$, $(5,7)$, and $(4,5)$, respectively, and all satisfy the hypotheses of Corollary~\ref{cor:llj-comparison}. Substitution into~\eqref{eq:llj-band} gives $d_{\mathrm{LLJ}}$. Every row satisfies $d_*>d_{\mathrm{LLJ}}$ and $d_*>q/2+1$. Thus each subfamily contains quantum MDS codes whose minimum distances exceed $q/2+1$ and are strictly larger than those of any code of the same length obtained from the linear-transformation constructions above, either directly or via the propagation rule.
\end{proof}

\subsection{Comparison with constructions based on cosets of multiplicative subgroups}

The original lengths in R5--R16 satisfy certain congruence conditions. Under the propagation rule, the reduction in length equals the reduction in minimum distance. These congruence conditions therefore restrict the minimum distance attainable at a prescribed target length. We use this observation to derive upper bounds and sufficient conditions for strict improvement. We then use upper bounds on the original lengths to show that neither the codes in R17--R20 nor codes derived from them via the propagation rule can have some of the lengths attained by our constructions. We first record the relation between the parameters before and after applying the propagation rule.

\begin{proposition}\label{prop:propagation-window}
Let a pure $q$-ary quantum MDS code have length $n_0$ and minimum distance $d_0$, where $d_0\le D_0$. If applying the propagation rule to this code yields a code of length $n$ and minimum distance $d$, then the following statements hold.
\begin{enumerate}[label=\textup{(\roman*)}]
\item The original length $n_0$ satisfies
\[
 n\le n_0\le n+D_0-d.
\]
\item Write $h_\pm=(q\pm1)/2$. If $D_0=q+1$ and there exist $h\in\{h_-,h_+\}$ and $r\in\{0,1\}$ such that $n_0\equiv r\pmod h$, then the minimum distance of the code obtained via the propagation rule satisfies
\begin{equation}\label{eq:congruence-propagation-bound}
 d\le q+1-\delta(q,n),\qquad
 \delta(q,n)=\min_{h\in\{h_-,h_+\},\,r\in\{0,1\}}
 \bigl((r-n)\bmod h\bigr),
\end{equation}
where $(r-n)\bmod h$ denotes the least nonnegative residue of $r-n$ modulo $h$.
\end{enumerate}
\end{proposition}

\begin{proof}
\textup{(i)} Let $\ell=n_0-n\ge0$ be the reduction in length under the propagation rule. The propagation rule and the inequality $d_0\le D_0$ give
\[
 d=d_0-\ell\le D_0-\ell.
\]
Thus $0\le n_0-n\le D_0-d$, which gives the stated range of lengths.

\textup{(ii)} Choose $h,r$ such that $n_0\equiv r\pmod h$. Since $n_0-n\ge0$ and $n_0-n\equiv r-n\pmod h$, we have
\[
 n_0-n\ge (r-n)\bmod h\ge\delta(q,n).
\]
The inequality in \textup{(i)} then gives the following bound on the minimum distance:
\[
 d\le q+1-(n_0-n)\le q+1-\delta(q,n).
\]
This proves the assertion.
\end{proof}

Families R5--R16 in the appendix summarize results from \cite{WangLuo2024,BarberoLucasEtAl2024,WanZhengZhu2025,WanLianZhu2025}. For each original code, the minimum distance is at most $q+1$, and the length satisfies one of the congruence conditions in Proposition~\ref{prop:propagation-window}\textup{(ii)}. Hence that proposition applies to the original codes (with $\ell=0$) and to all codes derived from them via the propagation rule. For our length $n=q^2-Sq+\theta$, the congruences $q\equiv1\pmod{h_-}$ and $q\equiv-1\pmod{h_+}$ give
\[
 n\equiv1-S+\theta\pmod{h_-},
 \qquad n\equiv1+S+\theta\pmod{h_+}.
\]
These two congruences express the upper bound on the minimum distance in terms of $q,S,\theta$, yielding the following corollary.

\begin{corollary}\label{cor:coset-explicit-improvement}
Suppose that an admissible set of parameters in Table~\ref{tab:general-comparison-parameters} satisfies $S\ge2$, and write $n=q^2-Sq+\theta$ and $d_*=q-S+1$. Let $d$ be the minimum distance of any code of length $n$ obtained from one of the families R5--R16, either directly or via the propagation rule. Then the following statements hold.
\begin{enumerate}[label=\textup{(\roman*)}]
\item If $\theta>S$ and $q>2\theta+4S+1$, then
\[
 d\le\frac{q+3}{2}+S+\theta,\qquad
 d_*-d\ge\frac{q-1}{2}-2S-\theta>0.
\]
\item If $-S<\theta\le-2$ and $q>4S$, then
\[
 d\le q-S+\theta+2,\qquad
 d_*-d\ge-\theta-1>0.
\]
\end{enumerate}
\end{corollary}

\begin{proof}
In \textup{(i)}, the inequalities $\theta>S$ and $q>2\theta+4S+1$ give
\[
 2\le1-S+\theta<h_-,\qquad
 2\le1+S+\theta<h_+.
\]
These two numbers are the least nonnegative residues of $n$ modulo $h_-$ and $h_+$, respectively. In the definition of $\delta(q,n)$, for each modulus the value corresponding to $r=0$ is smaller by $1$ than the value corresponding to $r=1$. Hence
\[
 \begin{aligned}
 \delta(q,n)
 &=\min\{h_--(1-S+\theta),\ h_+-(1+S+\theta)\}\\
 &=\frac{q-1}{2}-S-\theta,
 \end{aligned}
\]
where the second term is smaller than the first by $2S-1$.

In \textup{(ii)}, the inequalities $-S<\theta\le-2$ and $q>4S$ give
\[
 0<S-\theta-1<S-\theta<h_-,
 \qquad 2\le1+S+\theta<h_+.
\]
The first pair of inequalities gives $(-n)\bmod h_-=S-\theta-1$ and $(1-n)\bmod h_-=S-\theta$. Modulo $h_+$, the least nonnegative residue for $r=0$ is $h_+-(1+S+\theta)$, which is smaller by $1$ than the residue for $r=1$. Therefore,
\[
 \begin{aligned}
 \delta(q,n)
 &=\min\{S-\theta-1,\ h_+-(1+S+\theta)\}\\
 &=S-\theta-1.
 \end{aligned}
\]
Here the second term exceeds the first by $h_+-2S>0$. Substituting the respective values of $\delta(q,n)$ into Proposition~\ref{prop:propagation-window} gives the upper bounds on $d$ in \textup{(i)} and \textup{(ii)}. Subtracting these bounds from $d_*=q-S+1$ gives the corresponding lower bounds on the gain in minimum distance.
\end{proof}

The original codes listed in R17--R20 all have length at most $(q^2-1)/2$, and the same is true of every code derived from them via the propagation rule. For R17, this bound follows from the parameter conditions of Campion, Hernando, and McGuire \cite[Theorem~1.1]{CampionHernandoMcGuire2025}. The original length $n_0$ and the integer parameters $\lambda,\tau,\rho>1$ satisfy
\[
 n_0=\lambda\tau\sigma,\qquad
 \lambda\mid(q-1),\quad\tau,\rho\mid(q+1),\quad
 \gcd(\lambda,\tau)=1,
\]
and $2\le\sigma\le\rho/[\gcd(\lambda,\rho)\gcd(\tau,\rho)]$. Since $\gcd(\lambda,\tau)=1$, we have
\[
 \gcd(\lambda\tau,\rho)
 =\gcd(\lambda,\rho)\gcd(\tau,\rho).
\]
Consequently,
\[
 \begin{aligned}
 n_0
 &\le\frac{\lambda\tau\rho}{\gcd(\lambda,\rho)\gcd(\tau,\rho)}
 =\operatorname{lcm}(\lambda,\tau,\rho)\\
 &\le\operatorname{lcm}(q-1,q+1)
 =\frac{q^2-1}{2}.
 \end{aligned}
\]
The last equality follows because $q$ is odd. The original lengths in R18, R19, and R20 are at most $(q^2-1)/4$, $(q^2-1)/2$, and $(q^2-1)/2$, respectively. These bounds follow from the conditions $\mu+1\le\gamma/4$, $2\mu+1\le\gamma/2$, and the integer condition $m\ge2$, as listed in the appendix. Since the propagation rule does not increase the length, no code obtained from R17--R20, either directly or via the propagation rule, can have length $n>(q^2-1)/2$.

We again consider the two-parameter family of length $(q-R)(q-b)$ from the second construction. Applying Corollary~\ref{cor:coset-explicit-improvement}\textup{(i)} to this family gives a strict improvement over the minimum-distance bounds for the families R5--R16, including codes derived from them via the propagation rule, and provides an explicit choice of the required sets.

\begin{corollary}\label{cor:second-coset-family}
Let $R,b\ge2$ be integers, and write $S=R+b\ge5$. If $q$ satisfies
\begin{equation}\label{eq:second-coset-range}
 q>2Rb+4S+1,
\end{equation}
then, for every integer $2\le d\le d_*=q-S+1$, there exists a quantum MDS code $\qcode{n}{n-2d+2}{d}$ of length $n=(q-R)(q-b)$, with $d_*>q/2+1$. The value $d_*$ strictly exceeds the upper bound on the minimum distance of every code of length $n$ obtainable from the families R5--R16, either directly or via the propagation rule. For every such code of the same length, the difference between $d_*$ and its minimum distance is at least
\[
 \frac{q-1}{2}-Rb-2S>0.
\]
\end{corollary}

\begin{proof}
Let $\xi$ and $\rho$ be generators of $\mathbb F_q^{*}$ and $U$, respectively, and set
\[
 C=\{0\}\cup\{\xi^j:0\le j\le R-2\},
 \qquad D=\{\rho^j:0\le j\le b-1\}.
\]
By~\eqref{eq:second-coset-range}, we have $q>2S$. The exponents above are smaller than the orders of $\xi$ and $\rho$, respectively, so distinct exponents in each set yield distinct elements. Hence $|C|=R$ and $|D|=b$. Moreover, $b-1<(q+1)/2$ and $-1=\rho^{(q+1)/2}$ imply $D\subseteq U\setminus\{-1\}$. Set $u=R-1$. Then $u+b=S-1\le q-2$, so the sets $C$ and $D$ satisfy the hypotheses of the second construction. Corollary~\ref{cor:quantum-families} gives the stated length and minimum-distance range. Also, $q>2S$ implies $d_*>q/2+1$.

Since $R,b\ge2$ and $S\ge5$, we have
\[
 Rb-S=(R-1)(b-1)-1\ge1.
\]
Thus $\theta=Rb>S$. Together with~\eqref{eq:second-coset-range}, this allows us to apply Corollary~\ref{cor:coset-explicit-improvement}\textup{(i)}, which gives the upper bound $(q+3)/2+Rb+S$ on the minimum distance of every code of length $n$ obtained from R5--R16, either directly or via the propagation rule. Subtracting this bound from $d_*=q-S+1$ gives the stated lower bound on the gain. This lower bound is positive by~\eqref{eq:second-coset-range}, proving the strict improvement in minimum distance.

Furthermore, $q>2S$ gives
\[
 n=q^2-Sq+Rb>\frac{q^2}{2}>\frac{q^2-1}{2}.
\]
Thus no code from R17--R20, and no code derived from one of those families via the propagation rule, can have length $n$.

For fixed $R,b$ satisfying the stated conditions, inequality~\eqref{eq:second-coset-range} holds for all sufficiently large $q$. Hence we obtain a strict improvement in minimum distance at the same length for infinitely many values of $q$. Since $S$ and $Rb$ are then constant, the lower bound $q/2-(Rb+2S+1/2)$ on the gain in minimum distance grows linearly with $q$.
\end{proof}

\section{Conclusion}\label{sec:conclusion}

We have established sufficient conditions for constructing Hermitian self-orthogonal GRS codes whose locator sets are complements of unions of subsets of a finite field. These conditions determine the norms of the column multipliers by means of vanishing polynomials and Frobenius relations, while allowing the chosen subsets to intersect. Using cosets of multiplicative subgroups and sets with prescribed trace or norm values, we obtain five families of Hermitian self-orthogonal GRS codes and the corresponding quantum MDS codes. The last two families admit nonempty triple intersections. The five families of quantum codes have the common parameter form $n=q^2-Sq+\theta$ and $2\le d\le q-S+1$, where $S$ is determined by the cardinalities of the parameter sets and $\theta$ is determined jointly by these cardinalities and the intersection numbers.

This parameterization permits direct comparisons at the same length with previously known constructions and with codes derived from them via the propagation rule. We give sufficient conditions under which our codes have strictly larger minimum distances than those attainable from two trace-based constructions, and we determine the largest minimum distance attainable from the linear-transformation constructions in the range considered. For the second construction, we obtain a two-parameter family of length $(q-R)(q-b)$ together with conditions for strict improvement. For fixed $R$ and $b$ satisfying these conditions, the difference between our attained minimum distance and the largest minimum distance obtainable from the linear-transformation constructions, including all codes derived from them via the propagation rule, grows linearly with $q$. We also exhibit one infinite subfamily from each of the five constructions.

For the constructions based on cosets of multiplicative subgroups listed as R5--R16, length congruences yield explicit upper bounds on the minimum distances of both the original codes and all codes derived from them via the propagation rule, together with two sufficient conditions for strict improvement. In the range of Corollary~\ref{cor:second-coset-family}, the two-parameter family from our second construction satisfies the first condition. For fixed admissible $R$ and $b$, it gives strict improvements for all sufficiently large $q$, with a lower bound on the gain in minimum distance that grows linearly with $q$. The original lengths in R17--R20 are at most $(q^2-1)/2$, and the propagation rule does not increase the length; hence no code in these families, or derived from them via the propagation rule, can have the lengths covered by Corollary~\ref{cor:second-coset-family}.

Extending this method to extended GRS codes and to other structured subsets of finite fields with explicitly computable intersections is a natural direction for future work.

\begin{samepage}
\section*{Acknowledgements}
This work was supported by the Fundamental and Interdisciplinary Disciplines Breakthrough Plan of the Ministry of Education of China (JYB2025XDXM112) and the Science and Technology Commission of Shanghai Municipality (Grant No.~22DZ2229014).

\end{samepage}

\section*{Data availability}
No data were used in the research described in this article.

\section*{Declaration of competing interest}
The authors declare that they have no known competing financial interests or personal relationships that could have appeared to influence the work reported in this paper.

\appendix
\setcounter{table}{0}
\renewcommand{\thetable}{\Alph{section}.\arabic{table}}
\renewcommand{\theHtable}{appendix.\Alph{section}.\arabic{table}}
\renewcommand{\thetheorem}{\Alph{section}.\arabic{theorem}}
\clearpage
\newgeometry{left=18mm,right=18mm,top=18mm,bottom=18mm,footskip=9mm}
\clearpage
\thispagestyle{plain}
\noindent\makebox[\textwidth][c]{%
\rotatebox{90}{%
\begin{minipage}[c][\textwidth][c]{0.98\textheight}
\linespread{1}\selectfont
\setlength{\knownwidth}{\linewidth}
\setlength{\knownleft}{0.50\knownwidth}
\setlength{\knownright}{0.39\knownwidth}
\captionsetup{type=table,hypcap=false,font=small,skip=6pt}
\section{Selected parameters of previously known quantum MDS codes}\label{app:known}
{\small Tables~\ref{tab:recent-results} and~\ref{tab:earlier-results} summarize selected recent and earlier constructions, respectively. All parameters satisfy the remaining hypotheses of the cited theorems, and the distance $d$ is an integer with $d\ge2$.\par}
\centering
\caption{Selected recent families of quantum MDS codes (2024--2026, including publicly available preprints)}\label{tab:recent-results}
\begingroup
\linespread{1}\fontsize{8.5}{10.1}\selectfont
\renewcommand{\arraystretch}{1.12}
\setlength{\tabcolsep}{0pt}
\begin{tabular}{@{}l@{\hspace{7pt}}>{\raggedright\arraybackslash}p{\knownleft}@{\hspace{7pt}}>{\raggedright\arraybackslash}p{\knownright}@{\hspace{5pt}}r@{}}
\toprule
Family & Length $n$ and conditions & Distance range & Ref.\\
\midrule
R1 & $\textstyle n=\frac{(t_1+1)(q^2-1)}{h_1}+\frac{(2t_2+2)(q^2-1)}{h_2}-\frac{(t_1+1)(2t_2+2)(q^2-1)}{h_1h_2}$\par Conditions as in Theorem~8 of the cited paper & $\textstyle 2\le d\le\min\left\{\frac{(t_1+1)(q-1)}{h_1}-1,\left(\frac{h_2}{2}+t_2+1\right)\frac{q+1}{h_2}-2\right\}$ & \cite{TianLiWuChen2024}\\[2pt]
R2 & $\textstyle n=1+\mu\frac{q^2-1}{s}+\nu\frac{q^2-1}{t}-\frac{(q^2-1)\gcd(s,t)}{st}N_{\mu,\nu}$\par $s\mid(q-1)$, $t\mid(q+1)$, $1\le\mu\le s$, $1\le\nu\le t$, $\mu\nu<q-2$ & $\textstyle 2\le d\le\min\left\{\mu\frac{q-1}{s}+1,\left(\left\lfloor\frac t2\right\rfloor+1\right)\frac{q+1}{t}\right\}$ & \cite{FangWenFu2024}\\[2pt]
R3 & $n=(s+t)q-st,\ 1\le s,t\le q$ & $2\le d\le\min\{s,t\}+1$ & \cite{FangWenFu2024}\\[2pt]
R4 & $n=1+sq+t(q+1)-N_{s,t},\ 1\le s,t\le q-1$ & $2\le d\le\min\{s,t\}+1$ & \cite{FangWenFu2024}\\[2pt]
R5 & $\textstyle q=sm+1$\par $q,s$ odd, $n=\lambda\frac{q^2-1}{s}$, $1\le\lambda\le s$ & $\textstyle 2\le d\le\frac{s+1}{2}m+1$ & \cite{WangLuo2024}\\[2pt]
R6 & $\textstyle n=q^2-s(q+1)$\par $1\le s\le\frac{q-1}{2}$ & $2\le d\le q-s$ & \cite{WangLuo2024}\\[2pt]
R7 & $\textstyle n=\lambda(q+1)$\par $\lambda\mid (q-1)$ & $\textstyle 2\le d\le\frac{q+3}{2}$ & \cite{BarberoLucasEtAl2024}\\[2pt]
R8 & $\textstyle n=r\frac{q^2-1}{h}$\par $h\mid (q-1)$, $2\nmid h$, $1\le r\le h$ & $\textstyle 2\le d\le r\frac{q-1}{h}+1$ & \cite{WanZhengZhu2025}\\[2pt]
R9 & $\textstyle n=r\frac{q^2-1}{h}$\par $2\mid h\mid (q-1)$, $1\le r\le h/2$ & $\textstyle 2\le d\le\frac{q+1}{2}+r\frac{q-1}{h}$ & \cite{WanZhengZhu2025}\\[2pt]
R10 & $\textstyle n=r\frac{q^2-1}{h}$\par $2\mid h\mid (q-1)$, $h/2<r\le h$ & $\textstyle 2\le d\le\left\lfloor\frac{h+r}{2}\right\rfloor\frac{q-1}{h}+1$ & \cite{WanZhengZhu2025}\\[2pt]
R11 & $\textstyle n=1+r\frac{q^2-1}{h}$\par $h\mid (q+1)$, $1<r<\min\{q,h\}$, $2\nmid(r+h)$ & $\textstyle 2\le d\le\frac{r+h-1}{2}\frac{q+1}{h}$ & \cite{WanZhengZhu2025}\\[2pt]
R12 & $\textstyle n=r\frac{q^2-1}{h}$\par $2\mid h\mid (q+1)$, $1<r<h$, $2\nmid(r+h)$ & $\textstyle 2\le d\le\frac{r+h+1}{2}\frac{q+1}{h}-1$ & \cite{WanZhengZhu2025}\\[2pt]
R13 & $\textstyle n=1+r\frac{q^2-1}{2h}$\par $\frac{q-1}{h}=2\tau+1$, $h\in2\mathbb Z_{>0}$, $\tau\in\mathbb Z_{\ge1}$ & $\textstyle d\le\begin{cases}\frac{h+1}{2}\frac{q-1}{h}+\frac32,&\frac h2+1\le r\le h,\\[1mm]\frac{r(q-1)}{2h}+\frac32,&h<r<2h,\ r\text{ odd}
\end{cases}$ & \cite{WanLianZhu2025}\\[2pt]
R14 & $\textstyle n=r\frac{q^2-1}{2h}$\par $\frac{q-1}{h}=2\tau+1$, $h\in2\mathbb Z_{>0}$, $\tau\in\mathbb Z_{\ge1}$ & $\textstyle d\le\begin{cases}\frac{(h+2)(q-1)}{2h}+1,&\frac h2+1<r\le h,\ r\text{ odd},\\[1mm]\frac{r(q-1)}{2h}+\frac12,&h<r<2h,\ r\text{ odd}
\end{cases}$ & \cite{WanLianZhu2025}\\[2pt]
R15 & $\textstyle n=1+r\frac{q^2-1}{2h}$\par $\frac{q+1}{h}=2\tau+1$, $h\in2\mathbb Z_{>0}$, $\tau\in\mathbb Z_{\ge1}$ & $\textstyle d\le\begin{cases}\frac{(r+1)(q+1)}{2h},&h<r<\frac{3h}{2},\ r\text{ odd},\\[1mm]\frac{r(q+1)}{2h}-\frac12,&\frac{3h}{2}<r<2h,\ r\text{ odd}
\end{cases}$ & \cite{WanLianZhu2025}\\[2pt]
R16 & $\textstyle n=r\frac{q^2-1}{2h}$\par $\frac{q+1}{h}=2\tau+1$, $h\in2\mathbb Z_{>0}$, $\tau\in\mathbb Z_{\ge1}$ & $\textstyle d\le\begin{cases}\frac{(h+1)(q+1)}{2h}-\frac12,&\frac h2\le r\le h,\ r\text{ odd},\\[1mm]\frac{(r+2)(q+1)}{2h}-\frac12,&h<r<\frac{3h}{2},\ r\text{ odd},\\[1mm]\frac{(r+1)(q+1)}{2h}-1,&\frac{3h}{2}\le r<2h,\ r\text{ odd}
\end{cases}$ & \cite{WanLianZhu2025}\\[2pt]
R17 & $\textstyle n=\lambda\tau\sigma$\par $\lambda\mid(q-1)$, $\tau,\rho\mid(q+1)$, $\lambda,\tau,\rho>1$, $\gcd(\lambda,\tau)=1$, $2\le\sigma\le\rho/\kappa$, $\kappa=\gcd(\lambda,\rho)\gcd(\tau,\rho)$ & $\textstyle 2\le d\le T,\quad T=\begin{cases}(\lambda+4\tau)/2,&2\mid\lambda,\\ \lambda+\tau,&2\nmid\lambda\text{ and }(\lambda<\tau\text{ or }2\mid\tau\text{ or }\rho=2),\\ (\lambda+3\tau)/2,&2\nmid\lambda,\ \lambda>\tau,\ 2\nmid\tau,\ \rho\ne2
\end{cases}$ & \cite{CampionHernandoMcGuire2025}\\[2pt]
\bottomrule
\end{tabular}
\endgroup
\end{minipage}%
}}\par
\clearpage
\clearpage
\thispagestyle{plain}
\noindent\makebox[\textwidth][c]{%
\rotatebox{90}{%
\begin{minipage}[c][\textwidth][c]{0.98\textheight}
\linespread{1}\selectfont
\setlength{\knownwidth}{\linewidth}
\setlength{\knownleft}{0.50\knownwidth}
\setlength{\knownright}{0.39\knownwidth}
\captionsetup{type=table,hypcap=false,font=small,skip=6pt}
\centering
\ContinuedFloat
\caption{Selected recent families of quantum MDS codes (continued)}
\begingroup
\linespread{1}\fontsize{8.5}{10.1}\selectfont
\renewcommand{\arraystretch}{1.12}
\setlength{\tabcolsep}{0pt}
\begin{tabular}{@{}l@{\hspace{7pt}}>{\raggedright\arraybackslash}p{\knownleft}@{\hspace{7pt}}>{\raggedright\arraybackslash}p{\knownright}@{\hspace{5pt}}r@{}}
\toprule
Family & Length $n$ and conditions & Distance range & Ref.\\
\midrule
R18 & $\textstyle q\equiv-1\pmod4$\par $\gamma\mid2(q-1)$, $\gamma\nmid(q-1)$, $n=(\mu+1)\frac{q^2-1}{\gamma}$, $0\le\mu\le\gamma/4-1$ & $\textstyle 2\le d\le\left(\frac{\gamma+4}{8}+\mu\right)\frac{2(q-1)}{\gamma}+1$ & \cite{LiTianCaoLiu2025}\\[2pt]
R19 & $\textstyle 2\mid\gamma\mid (q+1)$\par $\gamma\equiv2\pmod4$, $n=(2\mu+1)\frac{q^2-1}{\gamma}$, $0\le\mu\le(\gamma-2)/4$ & $\textstyle 2\le d\le\frac{q+1}{2}+(\mu+1)\frac{q+1}{\gamma}-1$ & \cite{LiTianCaoLiu2025}\\[2pt]
R20 & $\textstyle n=\frac{q^2-1}{m}$\par $m=\frac{m_1m_2}{m_1+m_2-2}\in\mathbb Z_{\ge2}$, $\gcd(m_1,m_2)=2$, $m\nmid(q\pm1)$\par $m_1,m_2\in2\mathbb Z_{>0}$, $m_1\mid(q-1)$, $m_2\mid(q+1)$ & $d>q/2$ (for the parameters specified in the cited paper) & \cite{HeWangHuangChen2025}\\[2pt]
R21 & $\textstyle n=(q+1)r-\ell$\par $1\le r\le\frac{q-1}{2}$, $0\le\ell\le\frac{q-3}{2}+r$ & $\textstyle d=\frac{q+1}{2}+r-\ell$ & \cite{LiLiuJiang2026}\\[2pt]
R22 & $\textstyle n=(2r-1)q+1-\ell$\par $1\le r\le\frac{q-1}{2}$, $0\le\ell\le\frac{q-3}{2}+r$ & $\textstyle d=\frac{q+1}{2}+r-\ell$ & \cite{LiLiuJiang2026}\\[2pt]
R23 & $\textstyle n=2qr-\ell$\par $1\le r\le\frac{q-1}{2}$, $0\le\ell\le\frac{q-3}{2}+r$ & $\textstyle d=\frac{q+1}{2}+r-\ell$ & \cite{LiLiuJiang2026}\\[2pt]
\bottomrule
\end{tabular}
\endgroup

\vspace{8pt}
\captionsetup{type=table}
\caption{Selected earlier families of quantum MDS codes (published before 2024)}\label{tab:earlier-results}
\begingroup
\linespread{1}\fontsize{8.5}{10.1}\selectfont
\renewcommand{\arraystretch}{1.12}
\setlength{\tabcolsep}{0pt}
\begin{tabular}{@{}l@{\hspace{7pt}}>{\raggedright\arraybackslash}p{\knownleft}@{\hspace{7pt}}>{\raggedright\arraybackslash}p{\knownright}@{\hspace{5pt}}r@{}}
\toprule
Family & Length $n$ and conditions & Distance range & Ref.\\
\midrule
C1 & $2\le n\le q+1$ & $\textstyle 2\le d\le\left\lfloor\frac n2\right\rfloor+1$ & \cite{GrasslBethRoetteler2004,JinXing2014}\\[2pt]
C2 & $n=q^2+1$ & $2\le d\le q+1$ & \cite{FangFu2018}\\[2pt]
C3 & $n=tq,\ 1\le t\le q$ & $\textstyle 2\le d\le\left\lfloor\frac{tq+q-1}{q+1}\right\rfloor+1$ & \cite{FangFu2018}\\[2pt]
C4 & $n=t(q+1)+2,\ 1\le t\le q-1$ & $2\le d\le t+2$ & \cite{FangFu2018}\\[2pt]
C5 & $\textstyle n=1+r\frac{q^2-1}{h},\ h\mid (q-1),\ 1\le r\le h$ & $\textstyle 2\le d\le r\frac{q-1}{h}+1$ & \cite{FangFu2019}\\[2pt]
C6 & $\textstyle n=r\frac{q^2-1}{s}+\ell\frac{q^2-1}{t}-2r\ell\frac{q^2-1}{st}$\par $2\mid s\mid(q+1)$, $2\mid t\mid(q-1)$, $1\le r<s$, $1\le\ell\le t$, $st>2(q+1)$ & $\textstyle 2\le d\le\min\left\{r\frac{q+1}{s}-1,\frac{q+1}{2}+\frac{q-1}{t}-1\right\}$ & \cite{JinLuoFangQu2022}\\[2pt]
C7 & $\textstyle n=1+r\frac{q^2-1}{s+1}+\ell\frac{q^2-1}{t}-r\ell\frac{q^2-1}{(s+1)t}$\par $2\mid s$, $s+1\mid(q+1)$, $2\mid t\mid(q-1)$, $1\le r\le s$, $1\le\ell\le t$, $(s+1)t>q+1$ & $\textstyle 2\le d\le\min\left\{r\frac{q+1}{s+1},\frac{q+1}{2}+\frac{q-1}{t}-1\right\}$ & \cite{JinLuoFangQu2022}\\[2pt]
C8 & $\textstyle n=m(q-1)-\ell,\ 2\le m\le q,\ \ell\ge0$ & $\textstyle 2\le d\le\left\lfloor\frac{mq-1}{q+1}\right\rfloor+1-\ell$ & \cite{GuoLiLiu2021}\\[2pt]
C9 & $n=s(q+1)-\ell,\ 2\le s\le q-1,\ 0\le\ell\le s-2$ & $2\le d\le s-\ell$ & \cite{GuoLiLiu2021}\\[2pt]
C10 & $n=Aq-B,\ 0\le B<A<q$ & $2\le d\le A-B+1$ & \cite{LiXingWang2008}\\[2pt]
\bottomrule
\end{tabular}
\endgroup

{\linespread{1}\fontsize{8.5}{10.1}\selectfont\raggedright The first branch of R15 follows the range $h<r<3h/2$ in \cite[Theorem~3(1)]{WanLianZhu2025}; Table~2 of that paper instead gives $h<r\le3h/2$. Families R21--R23 are taken from \cite[Corollary~4.9]{LiLiuJiang2026}. The remaining admissibility conditions for the sets in R3 and R4, and the intersection conditions in R1, are those in the cited theorems; $N_{s,t}$ denotes the cardinality of the intersection of the two locator sets in R4. In R2, $N_{\mu,\nu}$ is defined as in \cite[Section~3]{FangWenFu2024}. For C8, the displayed distance interval is required to be nonempty.\par}
\end{minipage}%
}}\par
\clearpage
\restoregeometry

% Alphabetically ordered references, embedded for a self-contained source file.

\end{document}